\documentclass[onecolumn,amsmath,amssymb,floatfix,notitlepage,a4paper,prx,nofootinbib,superscriptaddress]{revtex4-2}

\usepackage{amssymb,amsthm,amsfonts,amstext,amsmath}
\usepackage{wrapfig}
\usepackage{float}
\usepackage{graphicx}
\usepackage{xcolor}
\usepackage{array}
\newcommand{\caphead}[1]{{\bf #1}}
\usepackage{tikz,pgfplots}\pgfplotsset{compat=1.18}

\usepackage{url}
\usepackage[colorlinks=true,citecolor=cyan,urlcolor=magenta]{hyperref}
\usepackage{cleveref}

\usepackage{tcolorbox}
\usepackage{enumitem}

\def\be{\begin{equation}}
\def\ee{\end{equation}}
\def\bea{\begin{eqnarray}}
\def\eea{\end{eqnarray}}
\def\bma{\begin{mathletters}}
\def\ema{\end{mathletters}}

\def\P{{\cal P}}

\def\G{{\cal G}}
\def\p{\overline{\pi}}
\def\q0{\underline{0}}

\def\H{{\cal H}}
\def\Z{{\cal Z}}
\def\T{{\cal T}}
\def\P{{\cal P}}
\def\p{\mathbb{P}}

\def\C{{\mathbb C}}
\def\id{{\mathbb I}}

\def\M{{\cal M}}
\def\H{{\cal H}}
\def\t{\mathbb{T}}

\def\B{{\cal B}}
\def\R{\mathbb{R}}

\def\D{{\cal D}}
\def\N{\mathbb{N}}

\def\tr{\mbox{tr}}

\def\X{{\cal X}}
\def\one{\leavevmode\hbox{\small1\normalsize\kern-.33em1}}

\def\bra#1{\langle#1|} \def\ket#1{|#1\rangle}
\def\braket#1#2{\langle#1|#2\rangle}

\def\proj#1{\ket{#1}\!\bra{#1}}

\def\id{{\mathbb I}}

\newtheorem{definition}{Definition}

\newtheorem{theorem}{Theorem}
\newtheorem{remark}{Remark}

\newtheorem{lemma}[theorem]{Lemma}
\newtheorem{prop}[theorem]{Proposition}

\newtheorem{example}{Example}

\begin{document}
\title{Non-commutative optimization problems with differential constraints}
\author{Mateus Ara\'ujo}
\affiliation{Departamento de Física Teórica, Atómica y Óptica, and Laboratory for Disruptive Interdisciplinary Science (LaDIS), Universidad de Valladolid, 47011 Valladolid, Spain}
\author{Andrew J. P. Garner}
\affiliation{Institute for Quantum Optics and Quantum Information (IQOQI) Vienna\\ Austrian Academy of Sciences, Boltzmanngasse 3, Wien 1090, Austria}
\author{Miguel Navascu\'es}
\affiliation{Institute for Quantum Optics and Quantum Information (IQOQI) Vienna\\ Austrian Academy of Sciences, Boltzmanngasse 3, Wien 1090, Austria}

\begin{abstract}
Non-commutative polynomial optimization (NPO) problems seek to minimize the state average of a polynomial of some operator variables, subject to polynomial constraints, over all states and operators, as well as the Hilbert spaces where those might be defined. Many of these problems are known to admit a complete hierarchy of semidefinite programming (SDP) relaxations. In this work, we consider a variant of NPO problems where a subset of the operator variables satisfies a system of ordinary differential equations. We prove that, under mild conditions of operator boundedness, for every such problem one can construct a standard NPO problem with the same solution. This allows us to define a complete hierarchy of SDPs to tackle the original differential problem. We apply this method to bound averages of local observables in quantum spin systems subject to a Hamiltonian evolution (i.e., a quench). We find that, even in the thermodynamic limit of infinitely many sites, low levels of the hierarchy provide very good approximations for reasonably long evolution times.
\end{abstract}

\maketitle

\section{Introduction}
The goal of non-commutative polynomial optimization (NPO) is to minimize the state average of a polynomial of operator variables, with the condition that a number of polynomials, evaluated on said variables, result in positive semidefinite operators. The minimization is understood to take place both over all states and tuples of operators and the Hilbert spaces that those inhabit \cite{NPO}. Many important problems in quantum physics can be cast as an NPO problem, e.g.: bounding the maximum quantum value of a Bell inequality \cite{NPA}, estimating the ground state energy of a local Hamiltonian \cite{Nakata2} and, more recently, bounding local observables of many-body systems at zero \cite{araujo2024firstorder} or finite temperature \cite{fawzi2023certified} or in a stationary state \cite{steady_state, steady_state2}. Nowadays NPO theory is a standard tool in quantum information theory \cite{Brown_2021,brown2023deviceindependentlowerboundsconditional,Moroder_2013}.

There exist, however, important physical problems that cannot be formulated within the framework of NPO theory. Note that the time evolution of a quantum system, open or closed, is modeled by means of differential equations on non-commuting variables. Since one cannot express them as polynomial inequalities, NPO theory cannot enforce such `differential constraints'. Thus, several tasks in quantum mechanics, such as predicting the future behavior of a many-body system \cite{PRXQuantum.4.020340, rand_hamil}, or mitigating the readout error of a Noisy Intermediate Scale Quantum (NISQ) device \cite{Takagi_2022} can only be accomplished in very special scenarios (e.g.: one-dimensional many-body systems) \cite{Vidal_2004}.


To tackle these challenges, we introduce the notion of \emph{differential NPO (DNPO) problems}. Such are optimization problems over non-commuting variables where a subset of the variables, indexed by a continuous number $t$, satisfies a system of ordinary differential equations. We show that each DNPO problem can be relaxed to an auxiliary NPO problem, which can therefore be approached with the usual semidefinite programming hierarchies. Moreover, if the resulting NPO problem is Archimedean (has bounded variables, so to speak), then its solution equals that of the original DNPO. Since Archimedeanity is also a sufficient condition for the NPO SDP hierarchies to converge to the solution of the auxiliary NPO problem, it follows that the proposed hierarchy of SDP relaxations is also complete for the original DNPO problem.



To illustrate the power of our method, we apply it to solve a very natural problem in statistical physics: the determination of the local properties of quenched many-body systems. More specifically, picture a finite spin system, initially in the ground or Gibbs state of some known local Hamiltonian $H$, and let $O$ be any local operator. We switch on the local Hamiltonian $H'$ for a time $\tau$. What is the new value of $\langle O\rangle$? We show that this problem can be cast as a DNPO. Its $\kappa^{th}$ SDP relaxation outputs the quantities $u^\kappa, l^\kappa$, which satisfy $l^\kappa\leq \langle O\rangle \leq u^\kappa$, with $\lim_{\kappa\to\infty} u^\kappa = \lim_{\kappa\to\infty} l^\kappa=\langle O\rangle$, if the initial state is unique. This problem can be generalized to the thermodynamic limit of infinitely many particles, where one would demand the initial state and Hamiltonians $H,H'$ to be translation-invariant. As in the case of finite size, we provide a complete hierarchy of SDP relaxations. We test low levels of the two hierarchies and find that the resulting upper and lower bounds are very close to each other for reasonably long times.

This paper is structured as follows: in section \ref{sec:def_DNPOs}, we will define the basic DNPO problem and motivate it by formulating a recent quantum information task as such. In section \ref{sec:hier_section} we will provide an NPO relaxation for the general DNPO problem, and derive a hierarchy of semidefinite programming \cite{sdp} relaxations to tackle the latter. We will also state a sufficient criterion for this hierarchy to converge to the solution of the original DNPO problem. The performance of the hierarchy of SDPs will be tested with two numerical examples in section \ref{sec:examples}. In section \ref{sec:quench} we show how our method can be adapted to tackle the problem of characterizing the average of many-body local observables after a quench. Finally, in section \ref{sec:conclusion} we will present our conclusions.


\section{Differential non-commutative polynomial problems}
\label{sec:def_DNPOs}
Let $Y\equiv(Y_1,\ldots,Y_N)$ be a vector of non-commutative Hermitian variables and let $X(t) := (X_1(t),\ldots,X_M(t))$ be a collection of vectors of non-commutative variables parameterized by $t\in[0,1]$. Since the boundary conditions $X(0), X(1)$ will play a privileged role and do not depend on $t$, we include them among the variables $Y$. That is, $Y_k:=X_k(0)$, for $k=1,\ldots,M$, and $Y_k:=X_{k-M}(1)$, for $k=M+1,\ldots,2M$. Given the polynomials $f(t,X(t),Y)$, $\{g_i(t,X(t),Y)\}_{i=1}^{M}$, $\{p_k(t, X(t),Y)\}_{k=1}^K$, and $\{q_j(t,X(t),Y)\}_{j=1}^J$ we consider the following optimization problem $\mathbf{P}$:
\begin{align}
\mathbf{P}:\qquad
\begin{array}{rl}
p^\star:=\min \;&\psi\left(\int_{[0,1]} dt f(t,X(t),Y)\right),\\[3pt]
\mbox{such that }&p_k(t, X(t),Y)\geq 0, \quad\mbox{for }k=1,\ldots,K,\; t\in[0,1],\\[3pt]
&\int_{[0,1]}dt\psi(q_j(t,X(t),Y))\geq 0,\quad\mbox{for }j=1,\ldots,J,\\[3pt]
&\frac{d}{dt}X_i(t)=g_i(t,X(t), Y),\quad\mbox{for }i=1,\ldots,M,\; t\in[0,1]
\end{array}
\label{problem}
\end{align}
\noindent where the optimization takes place over all Hilbert spaces $\H$, all operators $Y_i, X_j(t)\in B(\H)$ and all normalized states $\psi:B(\H)\to\C$. We call problem (\ref{problem}) a \emph{differential non-commutative polynomial optimization (DNPO) problem}. If $f(t,X(t),Y)=f(Y)$, the integral over $t$ is superfluous, and the objective function is simply the average of a polynomial on the operator variables $Y$. The same holds for the polynomials $\{q_j\}_j$. If $M=0$, i.e., if there are no $t$-dependent operators $X(t)$, this problem is a non-commutative polynomial optimization (NPO) problem \cite{NPO}.

Throughout the whole article, we will always assume that Problem (\ref{problem}) has feasible points, i.e., that the polynomials $\{p_k\}_k\cup\{q_j\}\cup\{g_i\}_i$ are such that there exist a Hilbert space $\H$, operators $\{X_i(t), Y_j\}_{ij}\subset B(\H)$ and a state $\psi:B(\H)\to\C$ satisfying the problem constraints.

To motivate the study of DNPO problems, we next show that a recently proposed task in quantum information theory can be formulated as such. The authors of \cite{Jones_2026} consider the problem of characterizing the correlations generated by measuring a quantum system of known spin number $J\in \frac{1}{2}\mathbb{N}$ with a device that allows rotating the system. Namely, the observed measurement statistics is of the form 
\begin{equation}
P(a|\vec{\theta})=\tr\{M_a U(\vec{\theta})\rho U(-\vec{\theta})\},
\label{rotation_box}
\end{equation}
where $(M_a)_{a=1}^A$ is a projective measurement, $\rho\in B(\H)$ is a quantum state and $U(\vec{\theta}):=e^{-i\sum_{k=1}^3\theta_kS_k}$, for some self-adjoint operators $S_1,S_2,S_3$ satisfying the relations:
\begin{align}
&\sum_{k=1}^3S_k^2=J(J+1)\id,\nonumber\\
&[S_1, S_2]=iS_3,[S_2, S_3]=iS_1,[S_3, S_1]=iS_2.
\end{align}
Neither the state nor the projective measurement nor the Hilbert space where they are defined are fixed. The set of all such ``rotation boxes'' is dubbed $Q_J$. For finitely-many vectors of angles $\zeta=\{\vec{\theta}^j\}_j$, deciding whether a finite collection of observed measurement statistics $\{P(a|\vec{\theta}^j)\}_{j=1}^m$ admits a representation of the form (\ref{rotation_box}) is an arduous task. In fact, the authors of \cite{Jones_2026} only manage to solve this problem for $A=2$ and two-angle scenarios $\zeta=\{(0,0,0), (0,0, \theta)\}$. 

Characterizing $Q_J$ is dual to conducting linear optimizations over this set. For some finite set of angles $\zeta=\{\vec{\theta}^j\}_{j=1}^m$, consider therefore the following problem:
\begin{align}
&\min_P \sum_{a=1}^A\sum_{j=1}^mf_{a,j}P(a|\vec{\theta}^j)\nonumber\\
\mbox{such that }&P(a|\vec{\theta})\in Q_J.
\end{align}
This problem can be reformulated as:
\begin{align}
&\min \psi\left(\sum_{a=1}^A\sum_{j=1}^mf_{a,j} U_j(1)^\dagger M_a U_j(1) \right)\nonumber\\
\mbox{such that }&\sum_{k=1}^3S_k^2=J(J+1)\id,\nonumber\\
&M_aM_b-M_a\delta_{ab}=0,a=1,...,A,\sum_a M_a=\id,\nonumber\\
&[S_1, S_2]=iS_3,[S_2, S_3]=iS_1,[S_3, S_1]=iS_2,\nonumber\\
&U_j(t)U_j(t)^\dagger=U_j(t)^\dagger U_j(t)=\id,j=1,...,m,\nonumber\\
&[U_j(t),\vec{\theta}^j\cdot \vec{S}]=0,j=1,...,m,\nonumber\\
&\frac{dU_j(t)}{dt}=iU_j(t)(\vec{\theta}^j\cdot \vec{S}),U_j(0)=\id,j=1,...,m.
\label{Markus}
\end{align}
Defining
\begin{align}
 X_j(t)&:=\frac{U_j(t)+U_j(t)^\dagger}{2},\mbox{ for } j=1,...,m,\nonumber\\
 X_j(t)&:=i\frac{U_j(t)-U_j(t)^\dagger}{2},\mbox{ for } j=m+1,...,2m,\nonumber\\
 Y_{4m+a}&:=M_a,\mbox{ for } a=1,...,A,\nonumber\\
 Y_{4m+A+k}&:=S_k,\mbox{ for } k=1,2,3,
\end{align}
the above problem is revealed to be a DNPO. If we could solve it or at least provide a good enough relaxation thereof, then we could use this knowledge to certify the randomness of the measurement outcome $a$, as the authors of \cite{Jones_2026} do in two-angle scenarios.

Notice also that
\begin{align}
&1+J(J+1)+3m-\sum_{j=1}^{2m}\left(X_j(t)^2+X_j(0)^2+X_j(1)^2\right)-\sum_a M_a^2-\sum_{k}S_k^2\nonumber\\
&=\frac{1}{2}\sum_{j=1}^m\left\{\left(1-U_j(t)U_j(t)^\dagger\right)+\left(1-U_j(t)^\dagger U_j(t)\right)+\left(1-U_j(0)U_j(0)^\dagger\right)+\left(1-U_j(0)^\dagger U_j(0)\right)\right.\nonumber\\
&\left.+\left(1-U_j(1)U_j(1)^\dagger\right)+\left(1-U_j(1)^\dagger U_j(1)\right)\right\}+\left(1-\sum_a M_a\right)+\left(\sum_a \left(M_a-M_a^2\right)\right)+ \left(J(J+1)\id-\sum_{k}S_k^2\right).
\end{align}
Due to the operator constraints, for any feasible point of Problem (\ref{Markus}) the right-hand side of the equation above vanishes. This implies, by the left-hand side, that any tuple of feasible operators $(X(t),Y)$ must be bounded. In the next section we will see that, for any DNPO with non-commuting variables whose norm can be bounded in this way, we can provide a complete hierarchy of SDP relaxations.

\section{A hierarchy of SDP relaxations for DNPO problems}
\label{sec:hier_section}

Problem (\ref{problem}) cannot be attacked by brute force. Note, indeed, that the optimization in (\ref{problem}) is over all Hilbert spaces $\H$, finite- or infinite-dimensional: a full parametrization of each operator appearing in (\ref{problem}) would therefore lead us to a problem with infinitely many real variables. In fact, some undecidable tasks can be reduced to the approximate resolution of the smaller class of NPO problems \cite{Fritz_2014}. Since there exist complete sets of relaxations for some such problems, it follows that there are no general and complete \emph{variational} algorithms to solve NPO problems and thus DNPO problems.

The previous argument does not preclude, however, the existence of complete sets of \emph{relaxations} for DNPO problems: finding these is the goal of this section. To achieve it, we take inspiration from the work of Lasserre et al. \cite{lasserre_ODE}. In this paper, it is shown that classical control problems can be reduced to polynomial optimization problems, which, in turn, can be tackled through hierarchies of convex relaxations \cite{lasserre, Parrilo2003}; more concretely, semidefinite programs (SDP) \cite{sdp}. 

Correspondingly, in the following lines we show how to relax problem $\mathbf{P}$ (\cref{problem}) to an NPO problem, for which there also exists a hierarchy of semidefinite programming (SDP) relaxations \cite{sdp}. This hierarchy is complete under the assumption that all the relevant operators are bounded~\cite{NPO}. This covers very interesting problems in quantum mechanics, such as the characterization of the set of rotation boxes (see problem (\ref{Markus}) in the previous section), the extrapolation of quantum measurement data (section \ref{sec:unitary}), and the computation of averages of local observables in many-body systems subject to a quench (section \ref{sec:quench}).

Let us then derive the promised NPO relaxation. Suppose, for some Hilbert space $\H$, that the operators $\{X_1(t), \ldots, X_M(t), Y_1, \ldots, Y_N\} \subset B(\H)$ satisfy the conditions of problem $\mathbf{P}$. On the extended Hilbert space $\tilde{\H}:= \H\otimes {\cal T}$, where ${\cal T}:={\cal L}^2[0,1]$ is spanned by the generalized vectors\footnote{Those are assumed to be orthogonal, i.e., $\braket{t}{u}=\delta(t-u)$.} $\{\ket{t}:t\in[0,1]\}$, we define the operators:
\begin{align}
\tilde{T} & :=\id_{\H}\otimes\int_{[0,1]} dt\;t\,\proj{t}_{{\cal T}}, \nonumber\\
\tilde{X}_i& :=\int_{[0,1]} dt\,X_i(t)\otimes\proj{t}_{{\cal T}},\nonumber\\
\tilde{Y}_j& := Y_j\otimes \id_{{\cal T}} \qquad \mathrm{for~}j =1,\ldots, N.
\label{eq:extended_operators}
\end{align}
We further define the vectors $\tilde{X} := (\tilde{X}_1,\ldots,\tilde{X}_M)$ and $\tilde{Y} := (\tilde{Y}_1,\ldots,\tilde{Y}_N)$. 

Note that, for any pair of functions $s_1, s_2:[0,1]\to B(\H)$, it holds that
\begin{equation}
\left(\int_{[0,1]} dt s_1(t)\otimes \proj{t}\right)\times \left(\int_{[0,1]} dt' s_2(t')\otimes \proj{t'}\right)=\int_{[0,1]} dt s_1(t)s_2(t)\otimes \proj{t}.
\end{equation}
For any polynomial $h(\tilde{T},\tilde{X},\tilde{Y})$, we thus have that
\begin{equation}
h(\tilde{T},\tilde{X},\tilde{Y})=\int_{[0,1]} dt\; h(t,X(t),Y)\otimes\proj{t}.   
\label{pol_corresp}
\end{equation}
For $k=1,\ldots,K$, the condition $p_k(t, X(t), Y)\geq 0$, for $t\in[0,1]$ is hence equivalent to the condition
\begin{equation}
p_k(\tilde{T}, \tilde{X}, \tilde{Y})\geq 0.
\label{p_corresp}
\end{equation}

Now, define the state $\omega\equiv \psi_{\H}\otimes \beta_{{\cal T}}$, where $\beta_{\cal T} \in B(\T)$ is any state assigning a uniform probability distribution over $t\in[0,1]$, e.g.: $\beta_{{\cal T}}(\bullet):=\int_{[0,1]^2} dtdt'\bra{t}\bullet\ket{t'}$.
From \cref{pol_corresp} it follows that
\begin{equation}
\omega\left( h(\tilde{T}, \tilde{X}, \tilde{Y})\right) = \int_{[0,1]}dt\;\psi\left(h(t,X(t),Y)\right).
\label{aver_corresp}
\end{equation}
Thus, we can re-express the objective function of problem $\mathbf{P}$ as
\begin{equation}
\omega\left(f(\tilde{T},\tilde{X},\tilde{Y})\right)=\psi\left(\int_{[0,1]} dt f(t,X(t),Y)\right),
\label{obj_corresp}    
\end{equation}
and all constraints of the form $\int_{[0,1]}dt\psi\left(q_j(t,X(t),Y)\right)\geq 0$ as 
\begin{equation}
\omega(q_j(\tilde{T},\tilde{X},\tilde{Y}))\geq 0.
\label{q_corresp}    
\end{equation}
From the definition of $\tilde{T}$, it follows that $\mbox{spec}(\tilde{T})=[0,1]$, which implies the extra operator inequality:
\begin{equation}
\tilde{T}-\tilde{T}^2\geq 0.
\label{t_corresp}
\end{equation}

It remains to enforce, by constraints on the new variables $\tilde{T}$, $\tilde{X}$, $\tilde{Y}$ and $\omega$, that $X(t)$ is the solution of the differential equation in problem $\mathbf{P}$. 
To do so, it will be convenient introduce a new set of asbtract, non-commutative variables: $\tilde{t},\tilde{x},\tilde{y}$. Those must not be interpreted as operators acting in a specific Hilbert space, but as the generators of an abstract $*$-algebra $\tilde{\p}$, or set of non-commutative polynomials.

Within the set of polynomials $\tilde{\p}$ we next define the `differentiation' linear map $\D:\tilde{\p}\to\tilde{\p}$ by induction on the degree of the monomial to which it is applied:
\begin{definition}
\label{def:differentiation}
The \textbf{differentiation} linear map $\D:\tilde{\p}\to\tilde{\p}$ satisfies
\begin{align}
\D(\tilde{t}) & = 1, \nonumber\\
\D(1) &=\D(\tilde{y}_k)=0 \quad\forall k,\nonumber\\
\D(\tilde{x}_i)&= g_i(\tilde{t},\tilde{x},\tilde{y}) \quad\forall i,\nonumber\\
\D(h h')& =\D(h)h'+h\D(h') \quad \forall h,h'\in \tilde{\P}.
\label{def_diff_map}
\end{align}
\end{definition}
Note that the differentiation map depends on the tuple of polynomials $g_1,...,g_M$ defining the differential equation in (\ref{problem}).

\begin{example}
Let $h_1=\tilde{t}^2\tilde{y}_8\tilde{x}_2\tilde{y}_9\tilde{x}_1+\tilde{x}_2(1)$. 
Then,
\begin{align}
\D h_1= 2\tilde{t}\tilde{y}_8 \tilde{x}_2 \tilde{y}_9 \tilde{x}_1 + \tilde{t}^2 \tilde{y}_8 g_2(\tilde{t},\tilde{x},\tilde{y}) \tilde{y}_9 \tilde{x}_1 + \tilde{t}^2 \tilde{y}_8 \tilde{x}_2 \tilde{y}_9 g_1(\tilde{t},\tilde{x},\tilde{y}).
\end{align}
\end{example}

If $\frac{d}{dt}X_i(t)=g_i(t,X(t), Y)$ for all $i$, then, for all polynomials $h$, it can be verified by induction on eq. (\ref{def_diff_map}) that:
\begin{align}
\label{eq:Dh}
(\D h)(\tilde{T}, \tilde{X}, \tilde{Y})=\int_{[0,1]} dt\,\frac{d}{dt}h(t,X(t),Y)\otimes\proj{t}.
\end{align}
By eq. (\ref{aver_corresp}), if we average the expression above on state $\omega$, we arrive at
\begin{align}
\omega\left(\D(h)(T,\tilde{X},\tilde{Y})\right)
& = \psi\left(\int_{[0,1]}dt\frac{d}{dt}h(t,X(t),Y)\right)=\psi\left(h(1,X(1),Y)\right)-\psi\left(h(0,X(0),Y)\right) \nonumber\\
& = \omega\left(h(1,\tilde{X}(1),\tilde{Y})\right) -\; \omega\left(h(0,\tilde{X}(0),\tilde{Y})\right),
\label{integral_conds}
\end{align}
where in the second equality we have invoked the fundamental theorem of calculus and the last equality follows from eq. (\ref{aver_corresp}) and the fact that the last two polynomials do not depend on $\tilde{T}$ or $\tilde{X}$.

Eqs. (\ref{p_corresp}), (\ref{obj_corresp}), (\ref{q_corresp}), (\ref{t_corresp}), and (\ref{integral_conds}) suggest the following relaxation $\mathbf{Q}$ of problem $\mathbf{P}$ (\cref{problem}):
\begin{align}
\mathbf{Q}:\qquad
\begin{array}{rl}
q^\star:=\min & \; \omega\left(f(\tilde{T},\tilde{X},\tilde{Y})\right),\\[2.5pt]
\mbox{such that }&[\tilde{T},\tilde{X}_i]=[\tilde{T},\tilde{Y}_j]=0\quad\forall i,j,\\[2.5pt]
&\tilde{T}-\tilde{T}^2 \geq 0, \\[2.5pt]
&p_k(\tilde{T}, \tilde{X}, \tilde{Y})\geq 0 \quad\mbox{for }k=1,\ldots,K,\\[2.5pt]
&\omega\left(q_j(\tilde{T},\tilde{X},\tilde{Y})\right)\geq 0 \quad\mbox{for }j=1,\ldots,J,\\
&\omega\left(\D(h)(\tilde{T},\tilde{X},\tilde{Y})\right) 
= \omega\left(h(\id,\tilde{X}(1),\tilde{Y})\right) -\;\omega\left(h(0,\tilde{X}(0),\tilde{Y})\right) \; \forall\, h\in \tilde{\p},
\end{array}
\label{problem_relax}
\end{align}
\noindent where the minimization is over the operators $(\tilde{T}, \tilde{X},\tilde{Y})$, the Hilbert spaces on which they act, and the normalized states $\omega$ in those spaces.

This is a non-commutative polynomial optimization (NPO) problem, as defined in \citet{NPO}, but with a countable (rather than finite) number of state constraints. Following \cite{NPO}, we next sketch how to derive a hierarchy of semidefinite programming (SDP) relaxations of problem (\ref{problem_relax}).

 The starting point to arrive at an SDP relaxation of (\ref{problem_relax}) is replacing the optimization over $\omega, \tilde{T}$, and $\tilde{Z}:=(\tilde{X},\tilde{Y})$ in Problem (\ref{problem_relax}) by an optimization over linear functionals $\tilde{\omega}:\tilde{\p}\to\C$. Intuitively, we wish $\tilde{\omega}$ to satisfy the relation
\begin{equation}
\tilde{\omega}(p(\tilde{t},\tilde{z})) = \omega(p(\tilde{T},\tilde{Z}))
\label{omega2omega}
\end{equation}
for all polynomials $p$.

How does relation (\ref{omega2omega}) constrain the functional $\tilde{\omega}$? For once, $\omega$ is normalized, so the relation above implies $\tilde{\omega}(1)=1$. Furthermore, $\omega$ is a state. Hence,
\begin{equation}
\tilde{\omega}(\mu\mu^\dagger)\geq0 \quad \forall \mu\in\tilde{\p}.
\label{omega_posi_def}
\end{equation}
Operator positivity constraints of the form $p(\tilde{T},\tilde{Z})\geq 0$ are modeled by demanding that
\begin{equation}
\tilde{\omega}(\mu p\mu^\dagger)\geq0  \quad \forall \mu\in\tilde{\p}.
\label{omega_ineq_def}
\end{equation}
Equality constraints of the form $h(\tilde{T},\tilde{Z})=0$ are similarly enforced by demanding $\tilde{\omega}$ to satisfy 
\begin{equation}
\tilde{\omega}(\mu h\nu)=0  \quad \forall \mu,\nu\in\tilde{\p}.
\label{omega_eq_def}
\end{equation}

Putting all together, we arrive at the problem:
\begin{align}
\begin{array}{rl}
\min & \; \tilde{\omega}\left(f(\tilde{t},\tilde{x},\tilde{y})\right),\\[2.5pt]
\mbox{such that }&\tilde{\omega}(1)=1,\tilde{\omega}(\mu\mu^\dagger)\geq 0,\forall \mu\in\tilde{\p},\\
&\tilde{\omega}(\mu [\tilde{t},\tilde{x}_i]\nu)=\tilde{\omega}(\mu [\tilde{t},\tilde{y}_j]\nu)=0\quad\forall i,j,\forall\mu,\nu\in\tilde{\p},\\[2.5pt]
&\tilde{\omega}(\mu(\tilde{t}-\tilde{t}^2)\mu^\dagger) \geq 0,\forall\mu\in\tilde{\p}, \\[2.5pt]
&\tilde{\omega}(\mu p_k(\tilde{t}, \tilde{x}, \tilde{y})\mu^\dagger)\geq 0 \quad\mbox{for }k=1,\ldots,K,\\[2.5pt]
&\tilde{\omega}\left(q_j(\tilde{t},\tilde{x},\tilde{y})\right)\geq 0 \quad\mbox{for }j=1,\ldots,J,\\
&\tilde{\omega}\left(\D(h)(\tilde{t},\tilde{x},\tilde{y})\right) 
= \tilde{\omega}\left(h(1,\tilde{x}(1),\tilde{y})\right) -\;\tilde{\omega}\left(h(0,\tilde{x}(0),\tilde{y})\right) \; \forall\, h\in \tilde{\mathbb{\p}}.
\end{array}
\label{problem_relax_functional}
\end{align}
In principle, this problem is just a relaxation of (\ref{problem_relax}). One can relax it further by placing limits on the degrees of the polynomials appearing above. More explicitly, for $k\in\N$, a sound relaxation of (\ref{problem_relax_functional}) is given by:
\begin{subequations}
\begin{align}
q^k:=\min & \; \tilde{\omega}\left(f(\tilde{t},\tilde{x},\tilde{y})\right),\\
\mbox{such that }&\tilde{\omega}(1)=1,\tilde{\omega}(\mu\mu^\dagger)\geq 0,\forall \mu\in\tilde{\p},\mbox{deg}(\mu)\leq k,\label{posi_SDP}\\
&\tilde{\omega}(\mu [\tilde{t},\tilde{x}_i]\nu)=\tilde{\omega}(\mu [\tilde{t},\tilde{y}_j]\nu)=0\quad\forall i,j,\forall\mu,\nu\in\tilde{\p},\deg(\mu)+\deg(\nu)\leq 2k-2,\label{vanish_comm}\\
&\tilde{\omega}(\mu(\tilde{t}-\tilde{t}^2)\mu^\dagger) \geq 0,\forall\mu\in\tilde{\p}, \deg(\mu)\leq k-1,\\
&\tilde{\omega}(\mu p_l(\tilde{t}, \tilde{x}, \tilde{y})\mu^\dagger)\geq 0 \quad\mbox{for }l=1,\ldots,K,\forall\mu\in\tilde{\p},\deg(\mu)\leq k-\left\lceil\frac{\deg(p_k)}{2}\right\rceil, \\[2.5pt]
&\tilde{\omega}\left(q_j(\tilde{t},\tilde{x},\tilde{y})\right)\geq 0 \quad\mbox{for }j=1,\ldots,J,\\
&\tilde{\omega}\left(\D(h)(\tilde{t},\tilde{x},\tilde{y})\right) 
= \tilde{\omega}\left(h(1,\tilde{x}(1),\tilde{y})\right) -\;\tilde{\omega}\left(h(0,\tilde{x}(0),\tilde{y})\right) \; \forall\, h\in \tilde{\p},\deg(h),\deg(\D(h))\leq 2k,
\label{SDP_diff}
\end{align}
\label{SDP_level_k}
\end{subequations}
where the optimization is over linear functionals $\tilde{\omega}$ defined on the set of polynomials in $\tilde{\p}$ of degree smaller than or equal to $2k$.

The above is a semidefinite program. To see why, first note that, due to the linearity condition, $\tilde{\omega}$ is fully specified by its value on the (finite) set of monomials of $\tilde{t},\tilde{z}$ of degree smaller than or equal to $2k$. That is, the number of optimization variables is finite. Second, again by linearity, to enforce conditions of the form (\ref{vanish_comm}) or (\ref{SDP_diff}), it is enough to consider $\mu,\nu,h$ to be monomials. Similarly, any positivity condition of the form 
\begin{equation}
\tilde{\omega}(\mu p\mu^\dagger)\geq0, \mbox{ for all } \mu\in \tilde{\p},\deg(\mu)\leq r,
\end{equation}
is equivalent to demanding the positive semidefiniteness of the matrix
\begin{equation}
M_{jk}(\omega,p,M):=\tilde{\omega}(o_j po_k^\dagger),
\end{equation}
where $\{o_j\}_j$ is the set of monomials of degree $r$ or smaller. The last matrix is called the \emph{localizing matrix} of the polynomial $p$. If $p=1$, it is called $\tilde{\omega}$'s moment matrix \cite{NPO}.

The sequence of SDP-computable values $(q^m)_m$ satisfies $q^1\leq q^2\leq\ldots\leq q^\star$. As shown in \cite{NPO}\footnote{In \cite{NPO}, only problems with a finite number of state constraints are discussed. However, the proof of primal convergence in \cite{NPO}[Theorem $1$] straightforwardly extends to scenarios with a countable number of state constraints.}, a sufficient condition for this sequence to converge to the solution $q^\star$ of problem $\mathbf{Q}$ is that the latter satisfies the Archimedean condition.

\begin{definition}[Archimedean condition / Sum of Squares]
\label{def:Archimedean}
Problem $\mathbf{Q}$ is said to satisfy the \textbf{Archimedean condition} if the polynomials 
\begin{equation}
\{e_k\}_k\ := \{p_j(\tilde{z})\}_j 
\cup \{\tilde{t}-\tilde{t}^2\}
\cup \{\pm i[\tilde{t},\tilde{x}_j]\}_j
\cup \{\pm i[\tilde{t},\tilde{y}_j]\}_j
\end{equation}
\noindent are such that
\begin{align}
&\lambda^2\id-\tilde{t}^2-\sum_{i} \tilde{x}_i^2-\sum_{j} \tilde{y}_j^2=\nonumber\\
&\sum_j s_j(\tilde{t},\tilde{z})^\dagger s_j(\tilde{t},\tilde{z})+\sum_{j,k} s_{jk}(\tilde{t},\tilde{z})^\dagger e_k(\tilde{t},\tilde{z})s_{jk}(\tilde{t},\tilde{z}),
\label{eq:Archimedean}
\end{align}
\noindent for some $\lambda\in \R$, and some polynomials $\{s_j\}_j$ and $\{s_{jk}\}_{j,k}$.

The right hand side of \cref{eq:Archimedean} is called a \textbf{sum of squares} (SOS) decomposition.
\end{definition}

Any polynomial $H(\tilde{t}, \tilde{z})$ admitting an SOS decomposition, evaluated on operators $\tilde{T},\tilde{Z}$ satisfying the constraints of problem $\mathbf{Q}$, results in a positive semidefinite operator. 
The Archimedean condition is essentially a condition on the operator positivity constraints that enforces feasible operators $\tilde{T}$, $\tilde{X}$ and $\tilde{Y}$ to be bounded. 

\begin{remark}
Any DNPO for which any set of feasible operators has norm bounded by $\gamma\in\R^+$ can be reformulated as an Archimedean DNPO, i.e., a DNPO leading to an Archimedean NPO relaxation. Indeed, if the constraints of the problem are such that, for any feasible point $(\H,\psi,X(t),Y)$, the relations $\|X_j(0)\|, \|X_j(1)\|, \|X_j(t)\|, \|Y_k\|\leq \gamma,\forall j,k$ hold, then we can add to the DNPO the redundant constraints:
\begin{align}
&\gamma^2-X_j(t)^2\geq 0,\mbox{ for } j=1,...,M,\nonumber\\
&\gamma^2-Y_k^2\geq 0,\mbox{ for } k=1,...,N.
\end{align}
The associated NPO will then be Archimedean, since
\begin{equation}
(M+N)\gamma^2+1-\tilde{t}^2-\sum_{i=1}^M\tilde{x}_i^2-\sum_{k=1}^N\tilde{y}^2_k=\sqrt{2}(\tilde{t}-\tilde{t}^2)\sqrt{2}+(1-\tilde{t})^2+\sum_{j=1}^{M+N}\left(\gamma^2-\tilde{z}_j^2\right).
\end{equation}

\end{remark}

Under the Archimedean condition, the hierarchy of SDP values $(q^m)_m$ thus converges to $q^\star$, the solution of Problem (\ref{problem_relax}). As it turns out, the Archimedean condition also implies that $q^\star=p^\star$. Hence, the proposed hierarchy of SDP relaxations (\ref{SDP_level_k}) is complete.

\begin{theorem}
\label{fund_theo}
Let problem $\mathbf{P}$ (\cref{problem}) be such that its NPO relaxation $\mathbf{Q}$ (\cref{problem_relax}) satisfies the Archimedean condition. Then, $p^\star=q^\star$. That is, the solutions of problems $\mathbf{P}$ and $\mathbf{Q}$ coincide.
\end{theorem}
\noindent The proof of this statement is presented in Appendix \ref{app:proof_theo}, and essentially follows by showing that, if the conditions of the theorem are satisfied, then any feasible point $(\tilde{\H},\tilde{X},\tilde{Y},\omega)$ of problem $\mathbf{Q}$ can be turned into a feasible point $(\H,X(t),Y,\psi)$ of $\mathbf{P}$ with the same objective value.

We conclude this section with two remarks: 
\begin{enumerate}
    \item Due to the constraints $[\tilde{t},\tilde{x}_i]=[\tilde{t},\tilde{y}_j]=0$, for any set of operators $\tilde{T},\tilde{Z}$ satisfying the constraints of problem \eqref{problem_relax}, it will hold that
\begin{equation}
(\tilde{T}-\tilde{T}^2)P_k(\tilde{T},\tilde{X}, \tilde{Y})\geq0 \quad \mathrm{for~} k=1,\ldots K.
\label{constraint_mateus}
\end{equation}
Including the corresponding SDP constraints in the hierarchy (\ref{SDP_level_k}) will not change its limiting value $q^\star$ (under the Archimedean condition). However, it might boost its speed of convergence.

\item The basic DNPO problem (\ref{problem}) can be extended to incorporate state constraints of the form:
\begin{align}
\psi\left(Q_j(t, X(t), Y)\right)\geq 0\quad\mbox{for }j=1,\ldots,J,\, t\in[0,1],
\label{eq:time_const}
\end{align}
Another interesting extension of problem $\mathbf{P}$ is a scenario where the vector variable $Y$ includes, not just the boundary values $X(0)$, $X(1)$, but $(X(t_1), \ldots,X(t_N))$, for $\{t_k\}_k\subset [0,1]$, with $t_1=0, t_N=1$. In Appendix \ref{app:extensions} we explain how to handle such extended DNPO problems.

\end{enumerate}

\subsection{Markovian DNPO problems}
\label{sec:Markovian}
Let $Y=(X(0),X(1),Y_{\text{rest}})$, and consider the following differential problem:
\begin{align}
\mathbf{P}:\qquad
\begin{array}{rl}
p^\star:=\min & \; \int_{[0,1]}dt\psi\left(f(t,X(t),Y_{\text{rest}})\right)+\psi(f^0(X(0),Y_{\text{rest}}))+\psi(f^1(X(1),Y_{\text{rest}})),\\[2.5pt]
\mbox{such that }&p^{0}_k(X(0), Y_{\text{rest}})\geq 0 \quad\mbox{for }k=1,\ldots,K^0,\\[2.5pt]
&p^{1}_k(X(1), Y_{\text{rest}})\geq 0 \quad\mbox{for }k=1,\ldots,K^1,\\[2.5pt]
&p^{[0,1]}_k(t, X(t), Y)\geq 0 \quad\mbox{for }k=1,\ldots,K^{[0,1]},\\[2.5pt]
&\int_{[0,1]}dt\psi\left(q^{[0,1]}_j(t,X(t),Y_{\text{rest}})\right)+\psi(q^0_{j}(X(0),Y_{\text{rest}}))+\psi(q^1_{j}(X(1),Y_{\text{rest}}))\geq 0,\\
&j=1,\ldots,J,\\
&\frac{dX(t)}{dt}=g(t,X(t),Y_{\text{rest}}).
\end{array}
\label{problem2}
\end{align}
Note that, save for the boundary conditions in the differential equation and the state average constraints, there are no relations between the sets of variables $\{t\}\cup\{X_i(t)\}_i$,$\{X_i(0)\}_i$ and $\{X_i(1)\}_i$. In particular, the constraints of the problem do not feature mixed products, such as $X_1(t)X_2(0)$. 

DNPO problems of this sort naturally appear in physical scenarios where one wishes to monitor the time evolution of a quantum system, see section \ref{sec:quench}. We call them \emph{Markovian DNPOs}.

Eq. (\ref{problem_relax_functional}) suggests us the following relaxation of problem (\ref{problem2}):
\begin{align}
\mathbf{Q}:\qquad
\begin{array}{rl}
q^\star:=\min & \; \tilde{\omega}\left(f(\tilde{t},\tilde{x},\tilde{y})\right)+\tilde{\omega}(f^0(\tilde{x}(0),\tilde{y}_{\text{rest}}))+\tilde{\omega}(f^1(\tilde{x}(1),\tilde{y}_{\text{rest}})),\\[2.5pt]
\mbox{such that }&\tilde{\omega}(1)=1,\tilde{\omega}(\mu \mu^\dagger)\geq 0,\forall\mu\in\tilde{\p},\\
&\tilde{\omega}(\mu p^{[0,1]}_k(\tilde{t}, \tilde{x}, \tilde{y})\mu^\dagger)\geq 0 \quad\mbox{for }k=1,\ldots,K^{[0,1]},\mu\in\tilde{\p},\\[2.5pt]
&\tilde{\omega}(\mu p^{0}_k(\tilde{x}(0), \tilde{y}_{\text{rest}})\mu^\dagger)\geq 0 \quad\mbox{for }k=1,\ldots,K^0,\mu\in\tilde{\p},\\[2.5pt]
&\tilde{\omega}(\mu p^{1}_k(\tilde{x}(1), \tilde{y}_{\text{rest}})\mu^\dagger)\geq 0 \quad\mbox{for }k=1,\ldots,K^1,\mu\in\tilde{\p},\\[2.5pt]
&\tilde{\omega}(\mu[\tilde{t},\tilde{x}_i]\nu)=\tilde{\omega}(\mu[\tilde{t},(\tilde{y}_{\text{rest}})_j]\nu)=0\quad\forall i,j,\forall\mu,\nu\in\tilde{\p},\\[2.5pt]
&\tilde{\omega}(\mu(\tilde{t}-\tilde{t}^2)\mu^\dagger)\geq 0,\forall\mu\in\tilde{\p}, \\[2.5pt]
&\tilde{\omega}\left(q^{[0,1]}_j(\tilde{t},\tilde{x},\tilde{y}_{\text{rest}})\right)+\tilde{\omega}(q^0_{j}(\tilde{x}(0),\tilde{y}_{\text{rest}}))+\tilde{\omega}(q^1_{j}(\tilde{x}(1),\tilde{y}_{\text{rest}}))\geq 0,j=1,\ldots,J,\\
&\tilde{\omega}\left(\D(h)(\tilde{t},\tilde{x},\tilde{y}_{\text{rest}})\right) 
= \tilde{\omega}\left(h(\id,\tilde{X}(1),\tilde{y}_{\text{rest}})\right) -\;\tilde{\omega}\left(h(0,\tilde{X}(0),\tilde{Y}_{\text{rest}})\right), \; \forall\, h\in \M.
\end{array}
\label{problem_relax2}
\end{align}
This formulation is, however, unnecessarily complicated. Why should $\tilde{\omega}$ be able to evaluate products of variables $\tilde{x},\tilde{x}(0),\tilde{x}(1)$, when those do not appear neither in the problem constraints, nor in the objective function? 

We next relax the problem above into an NPO that optimizes over three different linear functionals, namely, $\tilde{\omega}_{[0,1]}$, $\tilde{b}_0$ and $\tilde{b}_1$. Functional $\tilde{\omega}_{[0,1]}$ will only act on the set $\tilde{\p}_{[0,1]}$ of polynomials of $\tilde{t},\tilde{x}, \tilde{y}_{\text{rest}}$, while $\tilde{b}_0$ ($\tilde{b}_1$) will act on the set $\tilde{\p}_0$ ($\tilde{\p}_1$) of polynomial expressions of $\tilde{x}(0), \tilde{y}_{\text{rest}}$ ($\tilde{x}(1), \tilde{y}_{\text{rest}}$). 

With this notation, we propose the following problem:
\begin{align}
\mathbf{Q}':\qquad
\begin{array}{rl}
q^\star:=\min & \; \tilde{\omega}\left(f(\tilde{t},\tilde{x},\tilde{y})\right)+\tilde{b}_0(f^0(\tilde{x}(0),\tilde{y}_{\text{rest}}))+\tilde{b}_1\left(f^1(\tilde{x}(1),\tilde{y}_{\text{rest}})\right),\\[2.5pt]
\mbox{such that }&\tilde{\omega}_{[0,1]}(1)=1,\tilde{\omega}_{[0,1]}(\mu \mu^\dagger)\geq 0,\forall\mu\in\tilde{\p}_{[0,1]},\\
&\tilde{b}_0(1)=1,\tilde{b}_0(\mu \mu^\dagger)\geq 0,\forall\mu\in\tilde{\p}_0,\\
&\tilde{b}_1(1)=1,\tilde{b}_1(\mu \mu^\dagger)\geq 0,\forall\mu\in\tilde{\p}_1,\\
&\tilde{\omega}_{[0,1]}(\mu(\tilde{y}_{\text{rest}}))=\tilde{b}_0(\mu(\tilde{y}_{\text{rest}}))=\tilde{b}_1(\mu(\tilde{y}_{\text{rest}})),\forall\mu\in\tilde{\p}_{\text{rest}},\\
&\tilde{\omega}_{[0,1]}(\mu p^{[0,1]}_k(\tilde{t}, \tilde{x}, \tilde{y})\mu^\dagger)\geq 0 \quad\mbox{for }k=1,\ldots,K^{[0,1]},\mu\in\tilde{\p}_{[0,1]},\\[2.5pt]
&\tilde{b}_0(\mu p^{0}_k(\tilde{x}(0), \tilde{y}_{\text{rest}})\mu^\dagger)\geq 0 \quad\mbox{for }k=1,\ldots,K^0,\mu\in\tilde{\p}_0,\\[2.5pt]
&\tilde{b}_1(\mu p^{1}_k(\tilde{x}(1), \tilde{y}_{\text{rest}})\mu^\dagger)\geq 0 \quad\mbox{for }k=1,\ldots,K^1,\mu\in\tilde{\p}_1,\\[2.5pt]
&\tilde{\omega}_{[0,1]}(\mu[\tilde{t},\tilde{x}_i]\nu)=\tilde{\omega}_{[0,1]}(\mu[\tilde{t},(\tilde{y}_{\text{rest}})_j]\nu)=0\quad\forall i,j,\forall\mu,\nu\in\tilde{\p}_{[0,1]},\\[2.5pt]
&\tilde{\omega}_{[0,1]}(\mu(\tilde{t}-\tilde{t}^2)\mu^\dagger)\geq 0,\forall\mu\in\tilde{\p}_{[0,1]}, \\[2.5pt]
&\tilde{\omega}_{[0,1]}\left(q^{[0,1]}_j(\tilde{t},\tilde{x},\tilde{y}_{\text{rest}})\right)+\tilde{b}_0(q^0_{j}(\tilde{x}(0),\tilde{y}_{\text{rest}}))+\tilde{b}_1(q^1_{j}(\tilde{x}(1),\tilde{y}_{\text{rest}}))\geq 0,j=1,\ldots,J,\\
&\tilde{\omega}_{[0,1]}\left(\D(h)(\tilde{t},\tilde{x},\tilde{y}_{\text{rest}})\right) 
= \tilde{b}_1\left(h(\id,\tilde{x}(1),\tilde{y}_{\text{rest}})\right)-\tilde{b}_0\left(h(0,\tilde{x}(0),\tilde{y}_{\text{rest}})\right), \; \forall\, h\in \tilde{\mathbb{M}}_{[0,1]},
\end{array}
\label{problem_relax_relaxed}
\end{align}
where the minimum runs over all functionals $\tilde{\omega}_{[0,1]}$, $\tilde{b}_0$, $\tilde{b}_1$; $\tilde{\mathbb{P}}_{\text{rest}}$ denotes the set of polynomials on the variables $\tilde{y}_{\text{rest}}$; and $\tilde{\mathbb{M}}_{[0,1]}$ is the set of monomials on the variables\footnote{Abusing the notation, on the right-hand side of the last equation we apply them also to $\tilde{x}(1)$ and $\tilde{x}(0)$.} $\tilde{t},\tilde{x},\tilde{y}_{\text{rest}}$.

Like the sparse SOS hierarchies of \cite{Klep_2021} and the sequential hierarchy of \cite{klep2025quantitativequantumsoundnessbipartite}, relaxation (\ref{problem_relax_relaxed}) of Problem (\ref{problem_relax2}) exploits the sparsity in both objective and constraints to reduce the size of the moment and localizing matrices involved. Indeed, the implementation of the $k^{th}$ SDP relaxation of Problem (\ref{problem_relax_relaxed}) requires considerably less memory resources than its analog for problem (\ref{problem_relax2}). Very conveniently, although Problem (\ref{problem_relax_relaxed}) looks like a relaxed version of Problem (\ref{problem_relax2}), under mild boundedness conditions, its solution equals that of Problem (\ref{problem}). 

\begin{theorem}
\label{theorem:Markovian}
Denote by $\{\tilde{\chi}^{\alpha}_k\}_k$ the degree-$1$ monomials (namely, the variables) of $\tilde{\p}_{\alpha}$, with $\alpha\in\{[0,1],0,1\}$. Similarly, denote by $\{e_j^\alpha\geq 0\}_j$ all polynomial operator constraints associated to $\tilde{\p}^\alpha$. As long as the Archimedean condition holds for $\tilde{\p}_{[0,1]}$, $\tilde{\p}_{0}$, $\tilde{\p}_{1}$ independently, i.e., as long as there exists $\lambda\in \R^+$ such that
\begin{align}
&\lambda-\sum_{j} (\tilde{\chi}^\alpha_j)^2=\sum_j s_j^\dagger s_j+\sum_{j,k} s_{jk}^\dagger e^{\alpha}_ks_{jk},\mbox{ for some }\{s_j\}_j\cup\{s_{jk}\}_{jk}\subset \p_\alpha,
\label{indep_archimedean}
\end{align}
for $\alpha\in\{[0,1],0,1\}$, then Problems (\ref{problem_relax2}) and (\ref{problem_relax_relaxed}) are equivalent. Moreover, the associated hierarchy of SDP relaxations of (\ref{problem_relax_relaxed}), obtained by restricting the degree of the polynomials where the functionals $\tilde{\omega}_{[0,1]},\tilde{b}_0,\tilde{b}_1$ are defined, converges to the solution of the problem.
\end{theorem}
The reader can find a proof in Appendix \ref{app:Markovian}.

\section{Numerical tests}
\label{sec:examples}
We next test the performance of SDP hierarchy (\ref{SDP_level_k}) in a couple of concrete DNPO problems. For most of our numerical calculations, here and in the next section, we used the MATLAB optimization package YALMIP \cite{yalmip} in combination with the SDP solver Mosek \cite{mosek}.

\subsection{A simple non-commuting example}
Consider, for starters, the following (merely academic) problem:
\begin{equation}
\mathbf{P_2}:\quad
    \begin{array}{rl}
    \min_{\psi, Y_0, Y_1, \H} \quad \psi \left(\exp(Y_0Y_1+Y_1Y_0)\right) \\[2.5pt]
    \text{such that} \quad  {Y_0}^2 = Y_0, \quad {Y_1}^2=Y_1.
    \end{array}
    \label{P_2}
\end{equation}
If $Y_0$ and $Y_1$ were demanded to commute, the optimal value would be $1$. However, in the general noncommutative case, the optimal solution is in fact $e^{-1/4} \approx 0.7788 0078$\footnote{This can be seen by invoking Jordan's lemma \cite[Chapter VII]{bhatia2013matrix}, which allows us to express, modulo a unitary transformation, $Y_i=\bigoplus_j Y_i^j$, where $\{Y_i^j\}_j$ are $2\times 2$ matrices. By convexity, the solution of the above problem is therefore the result of an optimization over pairs of $2\times 2$ projector operators. Modulo unitary transformations, those can be parametrized by a single real number $\theta\in [0,2\pi)$, i.e., we can take 
\begin{equation}
Y_0=\left(\begin{array}{cc}1&0\\0&0\end{array}\right), Y_1=\left(\begin{array}{cc}\frac{1}{2}(1+\cos(\theta))&\frac{1}{2}\sin(\theta)\\\frac{1}{2}\sin(\theta)&\frac{1}{2}(1-\cos(\theta))\end{array}\right).
\end{equation}
A line search retrieves the stated result.}.

Suppose, though, that we ignored the solution of the problem. Introducing a family of operators $X(t)$ for $t\in[0,1]$, we can translate problem $\mathbf{P_2}$ into the equivalent DNPO $\mathbf{P_2}'$:
\begin{align}
\mathbf{P_2}':\quad
    \begin{array}{rl}
    \min_{\psi, X, Y_0, Y_1, \H}& \psi\left( X(1) \right) \\
    \text{such that} &  \frac{d}{dt}X(t) = (Y_0Y_1+Y_1Y_0) X(t), \\
    & Y_0^2 = Y_0, \quad Y_1^2=Y_1 \\
    & X(0) = \id \\
    \end{array}
    \label{P2_deformed}
\end{align}
Indeed, note that the only solution of the differential equation above is $X(t)=e^{(Y_0Y_1+Y_1Y_0)t}$; consequently, the objective functions of problems $\mathbf{P_2},\mathbf{P_2}'$ coincide.

Through the NPO reduction of section \ref{sec:hier_section}, we relax Problem $\mathbf{P_2}'$ to an NPO problem $\mathbf{Q_2}$, with operator variables $\tilde{T},\tilde{X},\tilde{X(1)},\tilde{Y}_1,\tilde{Y}_2$. To make the variables bounded, which is necessary for convergence, we add the additional constraints $\tilde{X} \ge 0$, $\tilde{X}(0) \ge 0$, $e^2\id-\tilde{X} \ge 0$, and $e^2\id-\tilde{X}(1) \ge 0$. ($Y_0$ and $Y_1$ are already bounded by virtue of being projectors, so there's no need to add analogous constraints for them). For $k=2$, we obtain the lower bound $0.4536 3658$; for $k=3$ we obtain $0.7788 0078$, which equals the problem solution within numerical precision\footnote{This might be an artifact of using $8$-digit precision: it is unlikely that the values of solution and relaxation exactly coincide.}.

\subsection{Interpolation and extrapolation under Hamiltonian evolution}
\label{sec:unitary}
We next consider an actual problem from quantum information theory: suppose that we initialize a quantum system, described through the Hilbert space $\H$, on some state $\psi$. We let the system evolve and, at time $\tau_k$, we measure the dichotomic\footnote{A dichotomic observable $A$ is a self-adjoint operator such that $A^2=\id$.} observable $A$. Repeating this experiment several times, we can estimate the average value $a(\tau_k)$ of observable $A$ at time $\tau_k$. We assume that the quantum system is closed, and therefore its dynamics are governed by a time-independent Hamiltonian $H\in B(\H)$, which we postulate to be bounded: more specifically, we assume that $E_{\max}\geq H\geq 0$, for some maximum energy value $E_{\max}>0$.

Basic quantum mechanics then tells us that
\begin{align}
a(\tau_k) := \psi \left( e^{i H \tau_k} A e^{-iH \tau_k} \right) = \psi \left( {U_k}^\dag A U_k \right),
\end{align}
where we have written $U_k := e^{-iH \tau_k}$.

For times $\{\tau_k\}_{k=1}^{N-1}$, we estimate the values $\{a(\tau_k)\}_k$. Given this information, we ask which range of values $a(\tau)$ could take for some other time $\tau\not\in\{\tau_k\}_k$. If $\tau$ is between the minimum and maximum values of $\{\tau_k\}_{k=1}^{N-1}$ this task amounts to an {\em interpolation} -- otherwise we speak of an {\em extrapolation}.

This task can be formulated as a pair of non-polynomial optimization problems $\mathbf{P_{U}}$:
\begin{align}
\mathbf{P_{U}}:\quad
\begin{array}{rl}
a^{\pm}(\tau) := \pm\max_{A, H, \mathcal{H}, \psi}  \quad &\pm\psi\left( e^{i H \tau} A e^{-iH \tau} \right) \\[2.5pt]
\text{such that} \quad & A = A^\dag, \; A^2 = \id, \; H = H^\dag, \\[2.5pt]
&0 \leq H \leq E_{\mathrm{max}} \id, \\[2.5pt]
&\psi\left( e^{i H \tau_k} A e^{-iH \tau_k}\right) = a(\tau_k) \quad \text{for~} k=1,\ldots, N\!-\!1.
\end{array}
\end{align}
If the energy assumption on the quantum system is sound, it then holds that $a(\tau)\in[a^{-}(\tau),a^{+}(\tau)]$.

By introducing the family of operators $\{U_k(t)\}_{k=1}^N$ for $t\in[0,1]$, we can reformulate the problem above as a DNPO problem $\mathbf{P_U}'$:
\begin{align}
\mathbf{P_{U}}':\quad
\begin{array}{rl}
     a^{\pm}(\tau) = \pm\max_{A, H, \{U_k\}_{k=1}^N, \pm\psi, \mathcal{H}}  & \psi\left( U_N(1)^\dag\,A\; U_N(1) \right) \\[2.5pt]
     \text{such that} & A = A^\dag, \; A^2 = \id, \; H = H^\dag, \\[2.5pt]
     &0 \leq H \leq E_{\mathrm{max}} \id, \\[2.5pt]
     &\psi\left( U_k(1)^\dag\,A\;U_k(1) \right) = a(\tau_k) \quad \text{for~} k=1,\ldots, N\!-\!1, \\[2.5pt]
     & \frac{d}{dt} U_k(t) = -i \tau_k H U_k(t) \quad \text{for~} k=1,\ldots, N, \\[2.5pt]
     &U_k(0) = \id \quad \text{for~} k=1,\ldots, N,
\end{array}
\label{extr_1}
\end{align}
where we set $\tau_N=\tau$. Indeed, as in the prior examples we see that $U_k(t) := e^{-iH\tau_k t}$ for each $k$ are the only set of solutions to the differential equations in $U_k$ satisfying the boundary conditions $U_k(0) = \id$. Since $H$ is Hermitian, these are obviously unitary operators (motivating our use of `$U$'). The associated NPO problem (\ref{problem_relax}), with variables $\tilde{U}_1,\ldots,\tilde{U}_N,\tilde{U}_1(1),\ldots,\tilde{U}_N(1), \tilde{H}$, can therefore be taken Archimedean, if we include the constraints
\begin{equation}
\tilde{U}_k\tilde{U}_k^\dagger=    \tilde{U}_k^\dagger\tilde{U}_k=\tilde{U}_k(1)\tilde{U}_k(1)^\dagger=    \tilde{U}_k(1)^\dagger\tilde{U}_k(1)=\id.
\end{equation}

The performance of the method is demonstrated in Fig.~\ref{fig:unitary_example} for interpolation and extrapolation with three sample points\footnote{The code for producing and solving this example this may be found at \url{https://github.com/ajpgarner/unitary_extrapolation}.} (i.e., $N=4$). Remarkably, the relaxation becomes exact at the sample points.

\begin{figure}[tbh]
\begin{centering}
\includegraphics[width=0.75\textwidth]{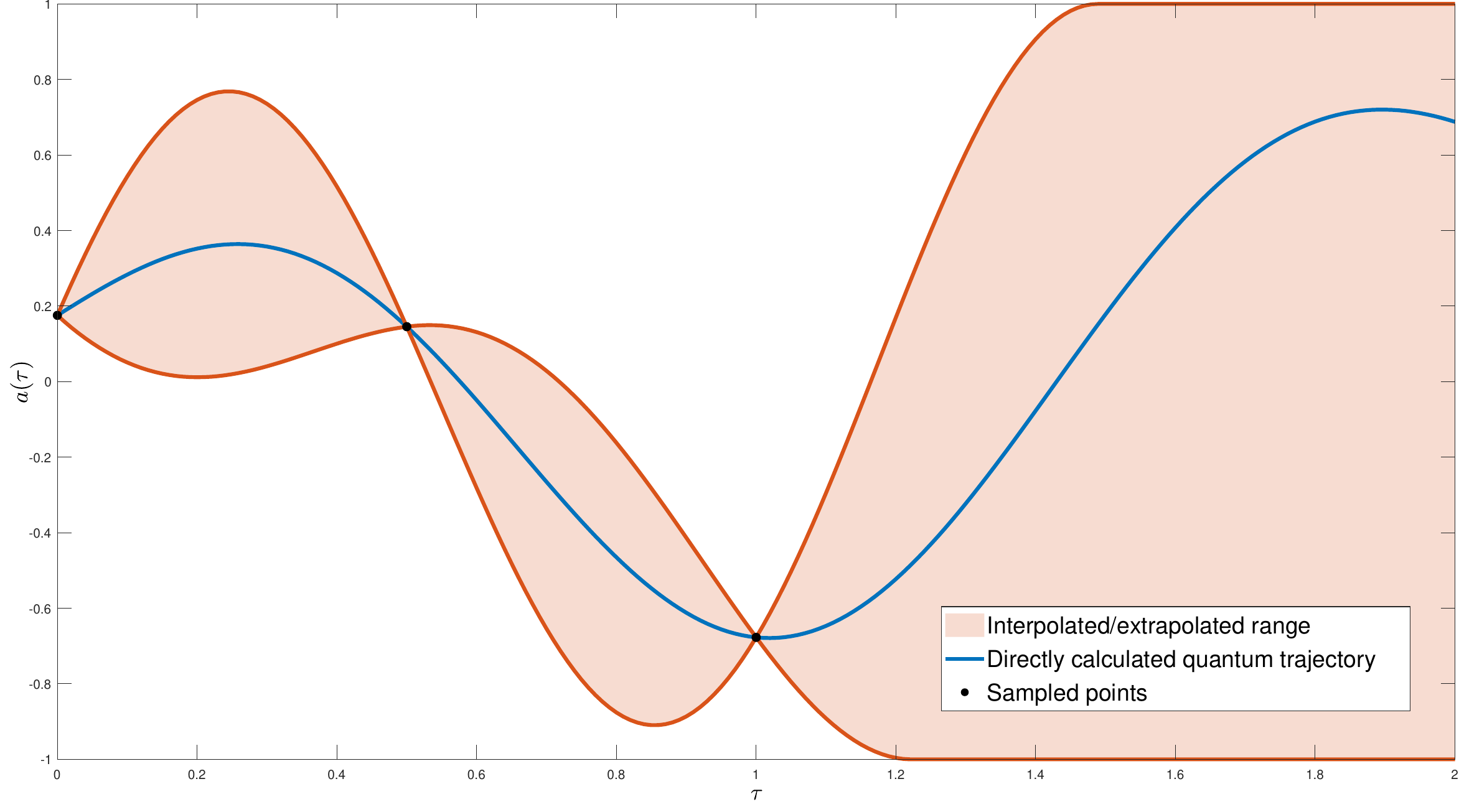}
\caption{\label{fig:unitary_example}
\caphead{Unitary interpolation and extrapolation.}
For $E_{\rm max} = 5$,
 the red region indicates the minimum and maximum values of $a(\tau)$ for different values of $\tau$, given the values of $a(\tau_1)$, $a(\tau_2)$, $a(\tau_3)$ at respective times $\tau_1=0$, $\tau_2=0.5$ and $\tau_3=1$, calculated using level $2$ moment matrices and localizing matrices. The blue line indicates a feasible quantum trajectory fitting the measurement data.
}
\end{centering}
\end{figure}

Alternatively, one can model the quantum extrapolation of time series as a DNPO problem with more than two boundary variables, see Appendix \ref{sec:multiple_boundaries}. Let us assume, for simplicity of notation, that the $\tau$'s are ordered, i.e., $0=\tau_1<\tau_2<\ldots<\tau_N$. Define $\bar{\tau}_k:=\frac{\tau_k}{\tau_{N}}$, and let $\tilde{k}\in\{1,...,N\}$ be such that $\tau_{\tilde{k}}=\tau$. Then we can reformulate the problem as:
\begin{align}
\begin{array}{rl}
     a^{\pm}(\tau) = \pm\max_{A, H, U, \psi, \mathcal{H}}  & \pm\psi\left( U\left(\bar{\tau}_{\tilde{k}}\right)^\dag\,A\; U\left(\bar{\tau}_{\tilde{k}}\right) \right) \\[2.5pt]
     \text{such that} & A = A^\dag, \; A^2 = \id, \; H = H^\dag, \\[2.5pt]
     &0 \leq H \leq E_{\mathrm{max}} \id, \\[2.5pt]
     &\psi\left(U\left(\bar{\tau}_k\right)^\dag\,A\;U\left(\bar{\tau}_k\right)\right) = a(\tau_k), \quad \text{for~} k\in\{1,\ldots, N\}\setminus\{\tilde{k}\}, \\[2.5pt]
     & \frac{d}{dt} U(t) = -i \tau_{N} H U(t),\\[2.5pt]
     &U(0) = \id.
\end{array}
\label{extr_2}
\end{align}

Compared to (\ref{extr_1}), the alternative formulation (\ref{extr_2}) of the extrapolation problem allows us, given the same computational resources, to take into account a higher number of experimental points. Fig. \ref{fig:7point_extrapolation} shows how the alternative method performs for $7$ sample points (namely, $N=8$).

\begin{figure}[tbh]
\begin{centering}
\includegraphics[width=0.75\textwidth]{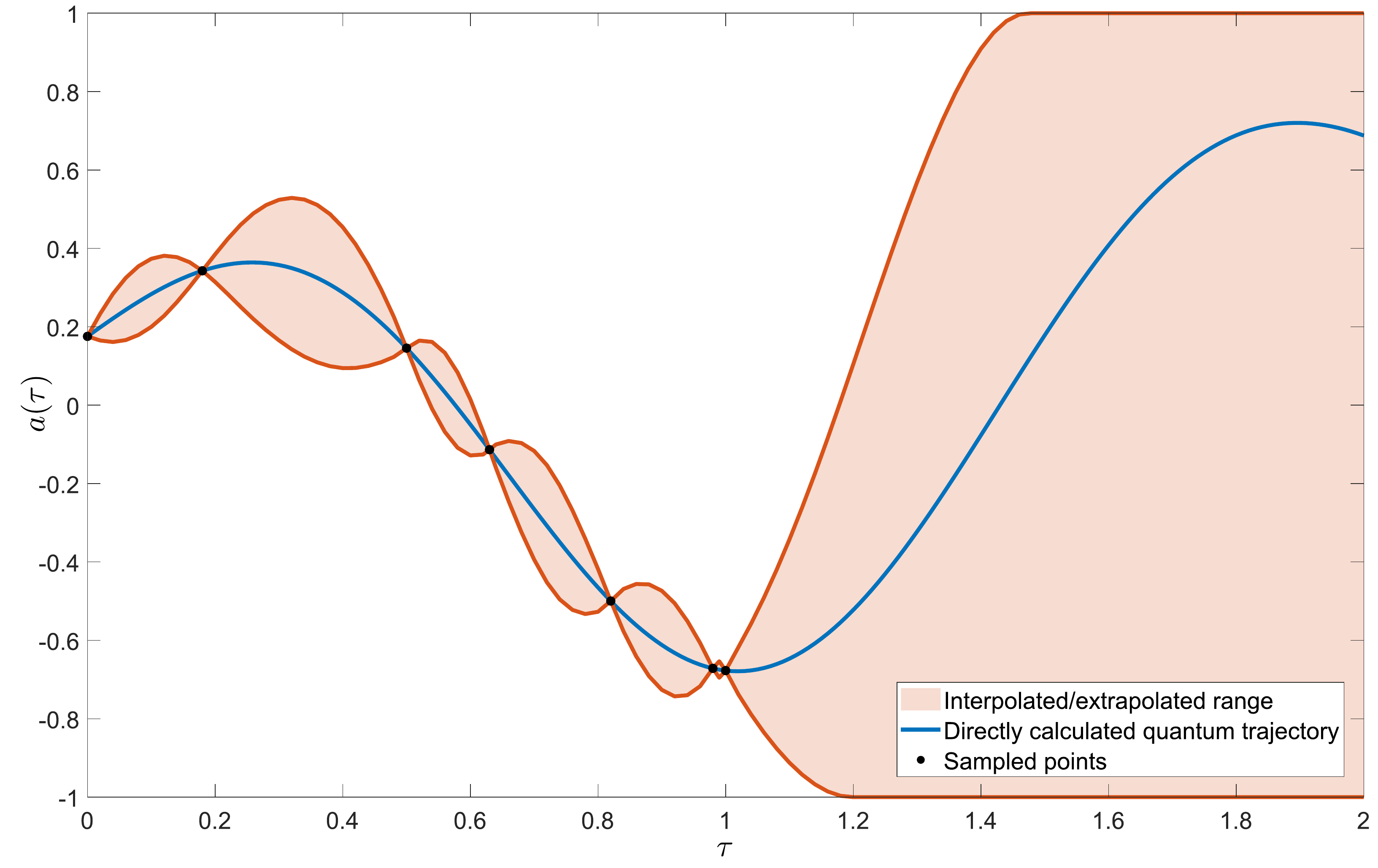}
\caption{%
\label{fig:7point_extrapolation}
\caphead{Unitary interpolation and extrapolation.}
For $E_{\rm max} = 5$, the red region indicates the minimum and maximum values of $a(\tau)$, for different values of $\tau$, given the data $\{a(\tau_k)\}_{k\not=\tilde{k}}$ observed at seven sample times, calculated using level $2$ moment matrices, level $2$ localizing matrices for $\tilde{H}$ and $\tilde{T}$ and level $1$ localizing matrices for the constraints in $\tilde{\Pi}_k$. The blue line indicates a feasible quantum trajectory that passes through all the sample points, indicated by black dots.}
\end{centering}
\end{figure}

\section{Application: Modeling quenches in 2-Reduced Density Matrix Theory}
\label{sec:quench}
\subsection{The quantum many-body scenario and 2-Reduced  Density Matrix Theory}

Consider a system of $n$ spins, let them be qubits for ease of exposition. Let $\sigma_c^{(j)}$, with $c\in\{1,2,3\}$ and $j\in\{1,\ldots,n\}$ denote the $c^{th}$ Pauli matrix acting on qubit $j\in\{1,\ldots,n\}$. Normally, it is assumed that the Hamiltonian $H$ that the system is subject to consists of sums of two-particle interactions, i.e.,
\begin{equation}
H=\sum_{j<k,a,b}h^{ab}_{jk}\sigma^{(j)}_a\sigma_b^{(k)}+\sum_{j,a} h_j^a\sigma^{(j)}_a,
\label{hamil_glass}
\end{equation}
where $\{h_{jk}^{ab}\}\subset \R$.
When interactions between every pair of particles are allowed, we speak of a \emph{spin glass}. In some physical scenarios, though, the spins lie on the sites of a regular lattice. In that case, only interactions between near neighbors are significant. For instance, in a 1D material with nearest-neighbor interactions, a general Hamiltonian would look like this:
\begin{equation}
H=\sum_{j,a,b}h^{a,b}_{j}\sigma^{(j)}_a\sigma_b^{(j+1)}+\sum_{j,a} h_j^a\sigma^{(j)}_a.
\label{hamil_1D}
\end{equation}
In quantum spin systems with a local (or few-particle) Hamiltonian, one seeks to answer questions like ``what is the ground state energy?" Or, ``what is the magnetization of the sample at temperature $T$?". 

The first question can be cast as an NPO problem. It reduces to
\begin{align}
&\min \psi(H)\nonumber\\
\mbox{such that } &\{\sigma^{(j)}_{a},\sigma^{(j)}_b\}=2\delta_{ab}\id,\forall a,b,j,\nonumber\\
&[\sigma^{(j)}_{a},\sigma^{(k)}_b]=2i\delta_{jk}\sum_c\varepsilon^{abc}\sigma^{(j)}_{c},\forall a,b,j,k,
\label{RDM_NPO}
\end{align}
where $\varepsilon^{ijk}$ is the Levi-Civita symbol. Indeed, the only irreducible representation of the above operator constraints is unitarily equivalent to the set of Pauli matrices acting on $n$ qubits.

Since the above is an NPO, one can use the method in \cite{NPO} to derive a complete hierarchy of SDP relaxations. This hierarchy, which can be shown to converge at level $k=n$ to the exact solution of the problem, is the basis of 2-Reduced Density Matrix Theory (2-RDMT) \cite{Nakata2,MazziottiRDMbook,plenio,BarthelAndHubener,bootstrapManyBody,bootstrapLawrence}.

Due to the commutation and anticommutation relations of the Pauli operators, polynomials of the non-commuting variables $\{\tilde{\sigma}^{(k)}_a\}_{a,k}$ can always be brought to a unique \emph{normal form}. More specifically, they can always be expressed as linear combinations of monomials of the form:
\begin{equation}
\mathbb{M}:=\left\{\prod_{j=1}^s\tilde{\sigma}^{(k_j)}_{a_j}:k_j\not=k_l, \mbox{ for } j\not=l\right\}.
\end{equation}
Note that all these monomials are self-adjoint. From now on, for any Pauli polynomial $p$, we denote by $[p]$ its normal form.

In principle, we could relax problem (\ref{RDM_NPO}) by optimizing over functionals $\tilde{\omega}$ defined on the span of $\{p\in \mathbb{M}:\mbox{deg}(p)\leq 2k\}$. For spin glasses, that would make perfect sense. For 1D, 2D or 3D materials, it is computationally preferable to tailor the considered polynomials to fit the neighboring relations between the $n$ qubits \cite{BarthelAndHubener}.

So, given a Hamiltonian $H$ like in (\ref{hamil_glass}), let $\G(H)$ be a graph with vertices $\{1,\ldots,n\}$, and edges 
\begin{equation}
\{(j,k):h^{ab}_{jk}\not=0, \mbox{ for some } a,b\}.   
\end{equation}
This graph represents the neighboring relations among the $n$ qubits, e.g.: for a spin glass, the edges of $\G(H)$ would be $\{(j,k):k,j=1,\ldots,n\}$; for a 1D material, $\{(j,j+1):j=1,\ldots,n-1\}$.

Next, define
\begin{equation}
\mathbb{M}(\G,k):=\left\{\prod_{j=1}^{k'}\tilde{\sigma}^{(s_j)}_{a_j}:s_j\not=s_l, \mbox{ for } j\not=l, \mbox{dist}_{\G}(s_j,s_l)< k, k'\leq k\right\}.
\end{equation}
With this notation, the $k^{th}$ SDP relaxation of problem (\ref{RDM_NPO}) is usually presented as:
\begin{align}
&\min \tilde{\omega}(H)\nonumber\\
\mbox{such that } &\tilde{\omega}(1)=1,\nonumber\\
&M\geq 0,\nonumber\\
&M_{o,o'}:=\tilde{\omega}([oo']),o,o'\in\mathbb{M}(\G(H),k),
\label{RDM_SDP_glass}
\end{align}
where the optimization takes place among linear functionals $\tilde{\omega}:\mathbb{M}(\G(H),k)^2\to\C$.

For several decades, the only rigorous quantities that 2-RDMT could provide were lower bounds on the ground state energy. In \cite{Wang_2024}, it is noted that, given an upper bound on the ground state energy (obtained, e.g., through variational methods), one can define an NPO to bound the value of the average of an arbitrary local operator within the set of ground states. More recently, it was discovered that, adding extra $H$-dependent convex constraints to the linear functional $\omega$, one can restrict the optimization to ground states of $H$ \cite{araujo2024firstorder,fawzi2023certified}, Gibbs states at any given temperature \cite{fawzi2023certified} and steady states of a local Lindbladian \cite{steady_state,steady_state2}. For any set of monomials ${\cal M}$, all such constraints have the form 
\begin{equation}
(\tilde\omega(o))_{o\in {\cal M}}\in G({\cal M}),
\end{equation}
which we abbreviate to $\tilde{\omega}({\cal M})\in G({\cal M})$. As the set ${\cal M}$ grows, the constraints become tighter, and $G({\cal M})$ exactly characterizes the desired set of states. In that regime, we simply write $\tilde{\omega}\in G$.

\subsection{Quenches in finite systems}

In the following lines, we detail how our results in Section \ref{sec:hier_section} imply that we can add a new functionality to the 2-RDMT toolbox, namely, a technique to model \emph{quenches}, or time-independent Hamiltonian evolutions. Suppose that the state of our quantum spin system initially belongs to some set of states $G$ (say, the ground state/Gibbs state of some local Hamiltonian $\texttt{H}$ or a steady state). We switch on Hamiltonian $H$ for time $\tau$, and consider the problem of lower bounding the value of the local operator $f(\sigma)$ on the new state. The problem to solve is therefore:
\begin{align}
&\min \psi(e^{iH(\sigma)\tau}f(\sigma)e^{-iH(\sigma)\tau})\nonumber\\
\mbox{such that }&\{\sigma^{(j)}_{a},\sigma^{(j)}_b\}=2\delta_{ab}\id,\forall a,b,j,\nonumber\\
&[\sigma^{(j)}_{a},\sigma^{(k)}_b]=2i\delta_{jk}\sum_c\varepsilon^{abc}\sigma^{(j)}_{c},\forall a,b,j,k,\nonumber\\
&\psi\in G.
\end{align}

This problem can be tackled by modeling the operator $e^{iH(\sigma)\tau}$, as we did in Section \ref{sec:unitary}). There is, however, a better way to deal with Hamiltonian evolution, which can be generalized to the thermodynamic limit of infinitely many spins.

First, we move to the Heisenberg picture. Namely, we introduce the time-dependent operator variables
\begin{gather}
\sigma^{(j)}_a(t):=e^{iHt\tau}\sigma^{(j)}_ae^{-iHt\tau}, \\
H(t):=e^{iHt\tau} H e^{-iHt\tau}.
\end{gather}
Then one finds that $H(t)=e^{iHt\tau}He^{-iHt\tau}=H$, with
\begin{align}
H(t)=\sum_{j<k}\sum_{a,b}h_{jk}^{a,b}\sigma_a^{(j)}(t)\sigma_b^{(k)}(t)+\sum_{j,a} h_j^a\sigma^{(j)}_a(t).
\end{align}
Thus, the variables $\sigma^{(j)}_a(t)$ follow the differential equation:
\begin{align}
&\frac{d\sigma^{(l)}_c(t)}{dt}=i\tau[H,\sigma^{(l)}_c(t)]=i\tau[H(t),\sigma^{(l)}_c(t)]=\nonumber\\
&-2\tau\left(\sum_{k>l}\sum_{a,b,d}h_{lk}^{a,b}\varepsilon^{acd}\sigma_d^{(l)}(t)\sigma_b^{(k)}(t)+\sum_{j<l}\sum_{a,b,d}h_{jl}^{a,b}\varepsilon^{bcd}\sigma_a^{(j)}(t)\sigma_d^{(l)}(t)+\sum_{a,d}h_l^a\varepsilon^{acd}\sigma_d^{(l)}(t)\right).
\end{align}
We arrive at the following problem:
\begin{align}
\min\; & \psi\left(f(\sigma)\right),\nonumber\\
\mbox{such that }& \{\sigma^{(j)}_{a}(t),\sigma^{(j)}_b(t)\}=2\delta_{ab}\id,\nonumber\\
&[\sigma^{(j)}_{a}(t),\sigma^{(k)}_b(t)]=2i\delta_{jk}\sum_c\varepsilon^{abc}\sigma^{(j)}_{c}(t),\\
&\frac{d\sigma^{(l)}_c(t)}{dt}=-2\tau\left(\sum_{k>l}\sum_{a,b,d}h_{lk}^{a,b}\varepsilon^{acd}\sigma_d^{(l)}(t)\sigma_b^{(k)}(t)+\sum_{j<l}\sum_{a,b,d}h_{jl}^{a,b}\varepsilon^{bcd}\sigma_a^{(j)}(t)\sigma_d^{(l)}(t)+\sum_{a,d}h_l^a\varepsilon^{acd}\sigma_d^{(l)}(t)\right),\nonumber\\
&\psi\in G.
\label{DNPO_spins}
\end{align}

This is a fairly general Hamiltonian simulation problem, which is widely believed to be intractable for classical computers. However, a relaxation of this problem would provide us with a lower bound $\mu_{-}$ on $\langle f(\sigma) \rangle_{\psi'}$, where $\psi'$ denotes the time-evolved state. By changing the sign of the objective function, one arrives at a problem whose relaxation would provide an upper bound $\mu_{+}$. Provided that $\mu_{+}-\mu_{-}\ll 1$, the problem can be considered solved for all practical purposes. 

Problem (\ref{DNPO_spins}) is a Markovian, independently Archimedean DNPO with the extra state constraints $\psi\in G$. The proof of Theorem \ref{theorem:Markovian} trivially extends to cover this case, so we can use the results of Section \ref{sec:Markovian} to reduce Problem (\ref{DNPO_spins}) to an optimization over three linear functionals, namely, $\tilde{\omega}_{[0,1]}$, $\tilde{b}_0$ and $\tilde{b}_1$. Functional $\tilde{\omega}_{[0,1]}$ will only act on the set $\tilde{\p}_{[0,1]}$ of polynomials of $\tilde{t},\tilde{\sigma}$, while $\tilde{b}_0$ ($\tilde{b}_1$) will act on the set $\tilde{\p}_0$ ($\tilde{\p}_1$) of polynomial expressions of $\tilde{\sigma}(0)$ ($\tilde{\sigma}(1)$). 

To simplify the presentation of the SDP relaxations, we note that, due to the Pauli relations and the commutation with $\tilde{t}$, the sets of polynomials $\p_{[0,1]},\p_0,\p_1$ respectively admit a normal form. Namely, for $k_\sigma,k_t$ large enough, any polynomial can be expressed as linear combinations of the self-adjoint monomials:
\begin{align}
&\mathbb{M}_{[0,1]}(\G,k_\sigma,k_t):=\left\{\tilde{t}^r\prod_{j=1}^s\tilde{\sigma}^{(k_j)}_{a_j}:k_j\not=k_l, \mbox{ for } j\not=l,\mbox{dist}_{\G}(k_j,k_l)< k_\sigma,s\leq k_\sigma,r\leq k_t\right\},\nonumber\\
&\mathbb{M}_{0}(\G,k_\sigma):=\left\{\prod_{j=1}^s\tilde{\sigma}^{(k_j)}_{a_j}(0):k_j\not=k_l, \mbox{ for } j\not=l,\mbox{dist}_{\G}(k_j,k_l)< k_\sigma,s\leq k_\sigma\right\},\nonumber\\
&\mathbb{M}_{1}(\G,k_\sigma):=\left\{\prod_{j=1}^s\tilde{\sigma}^{(k_j)}_{a_j}(1):k_j\not=k_l, \mbox{ for } j\not=l,\mbox{dist}_{\G}(k_j,k_l)< k_\sigma,s\leq k_\sigma\right\}.
\end{align}
Given a polynomial $p$, we will be dealing with moment and localizing matrices of the form
\begin{equation}
\Lambda_\alpha(p, \vec{k}):=\{(\Lambda_\alpha)_{o,o'}:=\omega_\alpha([opo']),o,o'\in \mathbb{M}_\alpha({\cal G}(H),\vec{k}\},
\end{equation}
for $\alpha\in\{[0,1],0,1\}$.

Let $\kappa:=(k_\sigma,k_t)$. With this notation, the $\kappa^{th}$ SDP relaxation of Problem (\ref{DNPO_spins}) is:
\begin{align}
q^{\kappa}:=\min & \; \tilde{b}_1\left(f(\tilde{\sigma}(1))\right),\nonumber\\
\mbox{such that }&\tilde{\omega}_{[0,1]}(1)=1,\Lambda_{[0,1]}(1,k_\sigma,k_t)\geq 0,\nonumber\\
&\Lambda_{[0,1]}(\tilde{t}-\tilde{t}^2,k_\sigma,k_t-1)\geq 0,\nonumber\\
&\tilde{b}_0(1)=1,\Lambda_0(1,k_\sigma)\geq 0,\nonumber\\
&\tilde{b}_0\in G\left(\left[\mathbb{M}_0(\G(H),k_\sigma)^2\right]\right),\nonumber\\
&\tilde{b}_1(1)=1,\Lambda_1(1,k_\sigma)\geq 0,\nonumber\\
&\tilde{\omega}_{[0,1]}\left(\left[j\tilde{t}^{j-1}o(\tilde{\sigma})+i\tau\tilde{t}^j[H(\tilde{\sigma}),o(\tilde{\sigma})]\right]\right) 
= \tilde{b}_1\left(o(\tilde{\sigma}(1))\right)-\delta_{j,0}\tilde{b}_0\left(o(\tilde{\sigma}(0))\right), \nonumber\\ 
&\forall o\in [\mathbb{M}_0(\G(H),k_\sigma)^2],\left[j\tilde{t}^{j-1}o(\tilde{\sigma})+i\tau\tilde{t}^j[H(\tilde{\sigma}),o(\tilde{\sigma})]\right]\in\mbox{span}\left(\left[\mathbb{M}_{[0,1]}(\G(H),k_\sigma,k_t)^2\right]\right).
\label{Markovian_NPO_rel}
\end{align}

Note that, if we restrict the optimization to functionals $\{\tilde{\omega}_j\}_{j=0,1}$, respectively defined in $\{\mathbb{M}_j(\G(H),k_\sigma,k_t)^2\}_{j=0,1}$, and functionals $\tilde{\omega}_{[0,1]}$, defined in $\mathbb{M}_{[0,1]}(\G(H),k_\sigma,k_t)^2$, then the size of the greatest positive semidefinite (PSD) constraint is $1+k_t$ times greater than the only PSD constraint needed to run 2-RDMT relaxations of order $k_\sigma$. 

To illustrate the performance of the hierarchy of SDPs (\ref{Markovian_NPO_rel}), we consider a quenching scenario where the initial many-body state is an element of the computational basis $\ket{a_1,\ldots,a_n}$, with $\vec{a}\in\{0,1\}^n$. The corresponding set of constraints on local momenta is in this case very simple, namely:
\begin{equation}
G_{\ket{\vec{a}}}=\{\langle\tilde{\sigma}_3^{(j)}\rangle=(-1)^{a_j},\forall j\}.    
\end{equation}
As a local observable, we choose $f(\sigma)=\sigma_3^{(n)}$, the spin of the last particle of the chain in the direction $\hat{z}$. We assume that the Hamiltonian $H$ is the Heisenberg Hamiltonian:
\begin{equation}
H = \frac{1}{4}\sum_{i=0}^{n-2} \sum_{a=1,2,3} \sigma_a^{(i)}\sigma_a^{(i+1)}.
\end{equation}
Taking $n=18$ and the initial state $\ket{1,0}^{\otimes 9}$, we set $k_\sigma=k_t=2$ in (\ref{Markovian_NPO_rel}) and compute upper and lower bounds for $\langle \sigma_3^{(18)}\rangle$, for $\tau\in [0,1]$, see Table \ref{tb:finite_size} \footnote{Our MATLAB codes for optimizations over 1D quantum systems can be found in \url{https://github.com/Navascues1980/1DEvol/tree/main}.}. We find that the greatest difference between upper and lower bounds is $0.0012$, achieved at $\tau=1$. This example illustrates that even low levels of hierarchy (\ref{Markovian_NPO_rel}) can provide good bounds for reasonably long evolution times.


        \begin{table}[ht]
        \begin{tabular}{c|ccc}
                $\tau$ & Lower bound &  Upper bound \\ \hline
  $0.1000$ & $0.9950$ & $0.9950$ \\
$0.2000$ & $0.9801$ & $0.9802$ \\
$0.3000$ & $0.9557$ & $0.9558$ \\
$0.4000$ & $0.9223$ & $0.9224$ \\
$0.5000$ & $0.8807$ & $0.8808$ \\
$0.6000$ & $0.8317$ & $0.8318$ \\
$0.7000$ & $0.7763$ & $0.7765$ \\
$0.8000$ & $0.7159$ & $0.7163$ \\
$0.9000$ & $0.6514$ & $0.6520$ \\
$1.0000$ & $0.5842$ & $0.5854$
        \end{tabular}
        \caption{\caphead{Local Hamiltonian evolution}. Lower and upper bounds for $\langle \sigma_3^{(18)}\rangle$ for several values of $\tau$.}
        \label{tb:finite_size}
        \end{table} 

When we tried to run system sizes greater than $n=18$ in our 256GB RAM desktop tower, the second-order SDP solver MOSEK \cite{mosek} reported an `out of memory' message. Therefore, we switched to the first-order solver SCS \cite{SCS1,SCS2,SCS3}, which allowed us to reach $n=25$ qubits in a desktop tower with 512 GB RAM. Taking as a starting state $\ket{1,0}^{\otimes 12}\otimes|1\rangle$, we also sought to compute the average $\langle\sigma^{(2)}_3\rangle$. The results for $k_\sigma=k_t=2$ are shown in the table below. The difference between upper and lower bounds is comparable to that of Table \ref{tb:finite_size}, which makes us conjecture that the approximation of the $k_\sigma=k_t=2$ SDP relaxation to the actual figures is essentially independent of $n$. It is worth remarking that SCS's actual memory use during the computations was just 15 GB. For system sizes beyond $n=25$, our program crashes at the pre-processing stage, from which we infer that MATLAB's poor allocation of memory resources is the main factor limiting the system size. Future implementations of our algorithm in other programming languages, such as Julia or C++, might allow reaching much larger system sizes with the same hardware.

        \begin{table}[ht]
        \begin{tabular}{c|ccc}
                $\tau$ & Lower bound &  Upper bound \\ \hline
  $0.1000$ & $-0.9950$ & $-0.9949$ \\
$0.2000$ & $-0.9800$ & $-0.9800$ \\
$0.3000$ & $-0.9557$ & $-0.9555$ \\
$0.4000$ & $-0.9225$ & $-0.9222$ \\
$0.5000$ & $-0.8809$ & $-0.8803$ \\
$0.6000$ & $-0.8320$ & $-0.8312$ \\
$0.7000$ & $-0.7768$ & $-0.7757$ \\
$0.8000$ & $-0.7165$ & $-0.7149$ \\
$0.9000$ & $-0.6523$ & $-0.6502$ \\
$1.0000$ & $-0.5857$ & $-0.5830$
        \end{tabular}
        \caption{\caphead{Local Hamiltonian evolution}. Lower and upper bounds for $\langle \sigma_3^{(25)}\rangle$ for several values of $\tau$ with SCS.}
        \label{tb:finite_size2}
        \end{table}

\subsection{Quenches in the thermodynamic limit}
We next consider a scenario where there are infinitely many spins in a regular lattice. We will assume, for simplicity, that the lattice is 1D, i.e., an infinite chain.

Systems in the thermodynamic limit of infinitely many particles are typically studied under the assumption of translation invariance. This means that both the initial state of the system and the Hamiltonian are invariant under the action of the \emph{translation operator} $\t$, an invertible homomorphism defined as:
\begin{equation}
\t (\sigma^{(j)}_a):= \sigma^{(j+1)}_a.
\end{equation}
In this context, a general translation-invariant Hamiltonian with nearest-neighbors interactions is written as:
\begin{equation}
H=\sum_{j\in \mathbb{Z}}\left(\sum_{a,b}h^{a,b}\sigma^{(j)}_a\sigma_b^{(j+1)}+\sum_{a} h^a\sigma^{(j)}_a\right).
\label{hamil_1D_TI}
\end{equation}
Note that the coefficients $h^{a,b}, h^a$ are site-independent.

The problem we want to solve is therefore:
\begin{align}
\min\; & \psi\left(f(\sigma(1))\right),\nonumber\\
\mbox{such that }& \{\sigma^{(j)}_{a}(t),\sigma^{(j)}_b(t)\}=2\delta_{ab}\id,j\in\mathbb{Z},a,b=1,2,3,\nonumber\\
&[\sigma^{(j)}_{a}(t),\sigma^{(k)}_b(t)]=2i\delta_{jk}\sum_c\varepsilon^{abc}\sigma^{(j)}_{c}(t),j,k\in\mathbb{Z},a,b=1,2,3,\\
&\frac{d\sigma^{(l)}_c(t)}{dt}=i\tau[H(\sigma(t)),\sigma_c^{(l)}(t)],\nonumber\\
&\psi\in G,\nonumber\\
&\psi\circ \t=\psi,
\label{DNPO_spins_TI}
\end{align}
Note that this time the index $j$ of the variables $\sigma^{(j)}_a(t)\}$ ranges from $-\infty$ to $\infty$. In addition, we are enforcing the initial state to be translation-invariant (TI). Notice as well that, despite the fact that $H(\tilde{\sigma})$ is not a polynomial, expressions of the form
\begin{equation}
[H(\tilde{\sigma}),p(\tilde{\sigma})]    
\end{equation}
are. Namely, for any polynomial $p$, the above equation represents a local observable.

If $\psi$ is TI, then so will be any state of the form $\psi_t(\bullet):=\psi(e^{iHt}\bullet e^{-i Ht})$. In the Heisenberg picture, this implies that
\begin{align}
&\psi(p(\t(\sigma(0))))=\psi(p(\sigma(0))), \forall p\in \tilde{\P}_0,\nonumber\\
&\psi(p(\t(\sigma(1))))=\psi(p(\sigma(1))), \forall p\in \tilde{\P}_1.
\end{align}
Following the derivation in Section \ref{sec:hier_section} of the NPO formulation of Problem (\ref{problem}), we find that the functionals $\tilde{\omega}_{[0,1]},\tilde{b}_0,\tilde{b}_1$ satisfy
\begin{align}
&\tilde{b}_0(p(\tilde{\t}(\tilde{\sigma}(0))))=\tilde{b}_0(p(\tilde{\sigma}(0))), \forall p\in \tilde{\P}_0,\nonumber\\
&\tilde{b}_1(p(\tilde{\t}(\tilde{\sigma}(1))))=\tilde{b}_1(p(\tilde{\sigma}(1))), \forall p\in \tilde{\P}_1,\nonumber\\
&\tilde{\omega}_{[0,1]}(p(\tilde{t},\tilde{\t}(\tilde{\sigma})))=\tilde{\omega}_{[0,1]}(p(\tilde{t},\tilde{\sigma})), \forall p\in \tilde{\P}_{[0,1]},
\label{TI_props}
\end{align}
where $\tilde{\t}$ is a linear map from polynomials to polynomials, defined by the relation
\begin{align}
&\tilde{\t}(p(\tilde{t},(\tilde{\sigma}^{(j)}_a)_{a,j},(\tilde{\sigma}^{(j)}_a(0))_{a,j},(\tilde{\sigma}^{(j)}_a(1))_{a,j})=p(\tilde{t},(\tilde{\sigma}^{(j+1)}_a)_{a,j},(\tilde{\sigma}^{(j+1)}_a(0))_{a,j},(\tilde{\sigma}^{(j+1)}_a(1))_{a,j}).
\end{align}

The last line of eq. (\ref{TI_props}) is proven as follows: 
\begin{align}
&\tilde{\omega}(\tilde{t}^kp(\tilde{\sigma}))=\int_{[0,1]}t^k dt\psi(e^{-iHt}p(\sigma(0))e^{iHt})=\hat{\psi}(p(\sigma(0))),
\label{interm_TI}
\end{align}
with $\hat{\psi}=\int dt t^k\psi_t$. Since $\psi_t$ is TI, so is $\hat{\psi}$. Hence,
\begin{align}
&\hat{\psi}(p(\sigma(0)))=\hat{\psi}\circ(p(\t(\sigma(0))))\nonumber\\
&=\tilde{\omega}(\tilde{t}^kp(\tilde{\t}(\tilde{\sigma})))=\tilde{\omega}\circ\tilde{\t}(\tilde{t}^kp).
\end{align}
Identifying the right-hand side of this equation with the left-hand side of eq. (\ref{interm_TI}) we arrive at the third line of eq. (\ref{TI_props}).

To model a thermodynamic state $\omega$ with finite computational resources, we consider its action on a finite set ${\cal O}$ of local operators and enforce that $\omega(o)=\omega(o')$, $\forall o,o'\in{\cal O}$ such that $o'=\t(o)$ (such `local' conditions are termed \emph{local translation invariance} (LTI) in the statistical physics literature \cite{LTI}). 

To handle the corresponding linear functionals, it will be useful to define the following sets of monomials:
\begin{align}
\mathbb{M}_{[0,1]}^{1D}(k_\sigma,k_t,n)&=\left\{\tilde{t}^{r}\prod_{j=1}^{s}\tilde{\sigma}^{(k_j)}_{a_j}:|k_j-k_i|< k_\sigma,k_i\not=k_j,\forall i\not=j,s\leq k_\sigma, \{k_j\}_j\subset\{1,\ldots,n\},r\leq k_t\right\},\nonumber\\
\mathbb{M}_{0}^{1D}(k_\sigma,n)&=\left\{\prod_{j=1}^{s}\tilde{\sigma}^{(k_j)}_{a_j}(0):|k_j-k_i|< k_\sigma,k_i\not=k_j,\forall i\not=j, s\leq k_\sigma,\{k_j\}_j\subset\{1,\ldots,n\}\right\},\nonumber\\
\mathbb{M}_{1}^{1D}(k_\sigma,n)&=\left\{\prod_{j=1}^{s}\tilde{\sigma}^{(k_j)}_{a_j}(1):|k_j-k_i|< k_\sigma,k_i\not=k_j,\forall i\not=j, s\leq k_\sigma,\{k_j\}_j\subset\{1,\ldots,n\}\right\}.
\end{align}
Intuitively, the sets above represent normal Pauli monomials generated by the first $n$ particles of the chain. The action of a translation operator over those will thus return a monomial with support on particles $\{2,\ldots,n+1\}$. Some of those will also have support in $\{2,\ldots,n\}$, which will allow us to enforce the LTI conditions on the state.

As before, we use some special symbols for the localizing matrices to lighten the notation:
\begin{equation}
\Lambda^{1D}_\alpha(p, \vec{k},n):=\{(\Lambda_\alpha)_{o,o'}:=\omega_\alpha([oo']),o,o'\in \mathbb{M}^{1D}_\alpha(\vec{k},n)\},
\end{equation}
for $\alpha\in\{[0,1],0,1\}$.

Once more, given a finite set of monomials ${\cal M}$, we denote by $G({\cal M})$ the set of convex constraints on the variables $(\omega(o):o\in {\cal M})$ that relax the property of $\omega$ being the ground or Gibbs state of some translation-invariant Hamiltonian $\texttt{H}$. See  \cite{araujo2024firstorder,fawzi2023certified} for complete sets of SDP or convex constraints characterizing the local averages of physically-relevant many-body states in the thermodynamic limit.

With the notation above, we define the SDP relaxation of order $\kappa=(k_\sigma,k_t,n)$ of problem (\ref{DNPO_spins_TI}) to be:
\begin{align}
q^{\kappa}:=\min & \; \tilde{b}_1\left(f(\tilde{\sigma}(1))\right),\nonumber\\
\mbox{such that }&\tilde{\omega}_{[0,1]}(1)=1,\Lambda^{1D}_{[0,1]}(1,k_\sigma,k_t,n)\geq 0,\nonumber\\
&\Lambda^{1D}_{[0,1]}(\tilde{t}-\tilde{t}^2,k_\sigma,k_t-1,n)\geq 0,\nonumber\\
&\tilde{\omega}_{[0,1]}(o)=\tilde{\omega}_{[0,1]}(o'),\forall o,o'\in \left[\mathbb{M}^{1D}_{[0,1]}(k_\sigma,k_t,n))^2\right], o'=\tilde{\t}^j(o),\mbox{ for some } j\in\mathbb{Z},\nonumber\\
&\tilde{b}_0(1)=1,\Lambda^{1D}_0(1,k_\sigma,n)\geq 0,\nonumber\\
&\tilde{b}_0\in G(\mathbb{M}^{1D}_0(\G(H),k_\sigma)^2),\nonumber\\
&\tilde{b}_0(o)=\tilde{b}_0(o'),\forall o,o'\in \left[\mathbb{M}^{1D}_{0}(k_\sigma,n))^2\right], o'=\tilde{\t}^j(o),\mbox{ for some } j\in\mathbb{Z},\nonumber\\
&\tilde{b}_1(1)=1,\Lambda_1(1,k_\sigma,n)\geq 0,\nonumber\\
&\tilde{b}_1(o)=\tilde{b}_1(o'),\forall o,o'\in \left[\mathbb{M}^{1D}_{1}(k_\sigma,n))^2\right], o'=\tilde{\t}^j(o),\mbox{ for some } j\in\mathbb{Z},\nonumber\\
&\tilde{\omega}_{[0,1]}\left(\left[j\tilde{t}^{j-1}o(\tilde{\sigma})+i\tau\tilde{t}^j[H(\tilde{\sigma}),o(\tilde{\sigma})]\right]\right) 
= \tilde{b}_1\left(o(\tilde{\sigma}(1))\right)-\delta_{j,0}\tilde{b}_0\left(o(\tilde{\sigma}(0))\right), \nonumber\\ 
&\forall o\in [\mathbb{M}^{1D}_0(k_\sigma,n)]^2,\left[j\tilde{t}^{j-1}o(\tilde{\sigma})+i\tau\tilde{t}^j[H(\tilde{\sigma}),o(\tilde{\sigma})]\right]\in\mbox{span}(\mathbb{M}^{1D}_{[0,1]}(k_\sigma,k_t,n)^2).
\label{Markovian_NPO_rel_TI}
\end{align}
Clearly, $q^\kappa\leq q^{\kappa'}$, for $\kappa\leq \kappa'$. This means, in particular, that the limit of the net $(q^\kappa)_{\kappa}$ exists.

Since the original DNPO problem (\ref{DNPO_spins_TI}) has infinitely many non-commuting variables, the convergence of the hierarchy (\ref{Markovian_NPO_rel_TI}) is not guaranteed by Theorem \ref{theorem:Markovian}. However, in Appendix \ref{app:thermo} we prove the following result:
\begin{theorem}
\label{theo:spins}
Define $q^\star:=\lim_{\kappa\to\infty}q^{\kappa}$. Then, there exists a Hilbert space $\H$, a representation $\pi:\p_0\to B(\H)$ and a TI state $\psi:B(\H)\to\C$, with $\psi\in G$, such that, for $\sigma_a^{(k)}:=\pi(\tilde{\sigma}_a^{(k)}(0))$, it holds that
\begin{align}
&\{\sigma^{(j)}_{a},\sigma^{(j)}_b\}=2\delta_{ab},\forall a,b,j,\nonumber\\
&[\sigma^{(j)}_{a},\sigma^{(k)}_b]=2i\delta_{jk}\sum_c\varepsilon^{abc}\sigma^{(j)}_{c},\forall a,b,j,k,\nonumber\\
\end{align}
and 
\begin{equation}
\psi(f(\sigma(1)))=q^\star,
\label{obj_spins}
\end{equation}
where $\sigma(t)$ is the solution of the differential equation
\begin{align}
&\frac{d\sigma_a^{(k)}(t)}{dt}=i\tau[H(\sigma(t)),\sigma_a^{(k)}(t)],\nonumber\\
&\sigma_a^{(k)}(0)=\sigma_a^{(k)}.
\label{diff_eq_spins}
\end{align}

\end{theorem}

To test the performance of the SDP hierarchy, we consider our initial state to be the ``all spins up" state $\psi\in G_{\ket{0}^{\otimes \infty}}$ and investigate what happens if we switch on the Hamiltonian
\begin{equation}
H=\sum_{j\in\mathbb{Z}}\sigma_1^{(j)}\sigma_2^{(j+1)}.
\end{equation}
Note that neither this Hamiltonian is reflection-invariant nor the local terms commute with each other. $H$ has appeared earlier in the context of entanglement detection in large 1D systems \cite{Wang_2017,entanglement_marginal}.

Figure \ref{fig:thermo} shows the result of applying the SDP relaxations $\kappa\in\{(2,2,7),(2,5,7),(2,2,10),(3,3,7)\}$ for different values of $\tau$. It is worth noting that, for the datasets corresponding to $\kappa\in\{(2,2,7),(2,5,7)\}$, which were generated by a 32GB desktop, upper and lower bounds are essentially undistinguishable up to $\tau\approx 0.7$. Hierarchy (\ref{Markovian_NPO_rel_TI}) thus provides very good results for short times, given reasonable computational resources.

\begin{figure}[tbh]
	\centering
 	\begin{tikzpicture}
		\begin{axis}[%
			scale only axis,
			xmin=0,
			xmax=2,
			ymin=-0.8,
			ymax=1,
			grid=major,
			xlabel={$\tau$},
            ylabel = {$\langle \sigma_3\rangle$},
			axis background/.style={fill=white},
			    every axis plot/.append style={thick}
			]
            \addplot[green!50!black, dashed] table[col sep=space] {n1022upp.txt};
            \addplot[green!50!black, dashed] table[col sep=space] {n1022low.txt};
            \addplot[cyan] table[col sep=space] {n733upp.txt};
            \addplot[cyan] table[col sep=space] {n733low.txt};
            \addplot[black, dashed] table[col sep=space] {n722upp.txt};
            \addplot[black, dashed] table[col sep=space] {n722low.txt};            
            \addplot[red!70!black, dotted] table[col sep=space] {n725upp.txt};
            \addplot[red!70!black, dotted] table[col sep=space] {n725low.txt};            

        \end{axis}
	\end{tikzpicture}
    \caption{\label{fig:thermo}
\caphead{Thermodynamic limit.}
The plot indicates upper and lower bounds for $\langle\sigma_3\rangle$ at time $\tau$, calculated with the SDP relaxation (\ref{Markovian_NPO_rel_TI}), with $\kappa=(2,2,7)$ (dashed black), $\kappa=(2,5,7)$ (dotted red), $\kappa=(2,2,10)$ (dashed green) and $\kappa=(3,3,7)$ (solid blue). As the curves attest, increasing a bit the values of both $k_\sigma$ and $k_t$ significantly boosts predictability for larger times $\tau$.}
\end{figure} 

\subsection{Higher local and spatial dimensions}
In the last sections, we have been assuming that our many-body system is made of $n$ (or infinitely many) qubits. Our results, however, easily extend to condensed matter systems with higher local dimensions. Suppose, indeed, that we wished to deal with condensed matter systems of local dimension $D$. The first step would be finding a basis of Hermitian operators in $B(\C^D)$. One can take, e.g., the generators of $SU(D)$ in dimension $D$. A simpler alternative is using the operators
\begin{align}
&\id,\nonumber\\
&\bar{P}_{j}:=\proj{k},k=1,\ldots,n-1,\nonumber\\
&\bar{X}_{jk}:=\ket{k}\bra{j}+\ket{j}\bra{k},j,k\in\{1,\ldots,n\},j>k,\nonumber\\
&\bar{Y}_{jk}:=i(\ket{k}\bra{j}-\ket{j}\bra{k}),j,k\in\{1,\ldots,n\},j>k,
\end{align}
with $j>k$, $j,k=1,\ldots,D$. One can uniquely express any product of operators in ${\cal O}_D:=\{\bar{X}_{jk},\bar{Y}_{jk},\bar{P}_{j}\}_{j,k}$ as a linear combination of elements of ${\cal O}_D$. If we define one such set ${\cal O}^l_D:=\{\bar{X}^{(l)}_{jk},\bar{Y}^{(l)}_{jk},\bar{P}^{(l)}_{j}\}_{j,k}$ of variables for each subsystem $l$, i.e.,
\begin{align}
&\bar{P}^{(l)}_{j}:=\id^{\otimes (l-1)}\otimes\bar{P}_k\otimes \id^{\otimes (n-l)},k=1,\ldots,n-1,\nonumber\\
&\bar{X}^{(l)}_{jk}:=\id^{\otimes (l-1)}\otimes X_{jk}\otimes \id^{\otimes (n-l)},,j,k\in\{1,\ldots,n\},j>k,\nonumber\\
&\bar{Y}^{(l)}_{jk}:=\id^{\otimes (l-1)}\otimes Y_{jk}\otimes \id^{\otimes (n-l)},j,k\in\{1,\ldots,n\},j>k,
\end{align}
then again any product of those can be expressed as linear combinations of tensor products of operators in ${\cal O}_D$. Thus, like in the qubit case, any polynomial of the corresponding abstract, non-commuting variables $\{x^{(l)}_{jk},y^{(l)}_{jk},p^{(l)}_{j}\}$ admits a unique normal form. Crucially, for any irreducible representation $\pi$ satisfying the constraint
\begin{equation}
\mu(\pi(x),\pi(y),\pi(z))=[\mu](\pi(x),\pi(y),\pi(z)),\forall \mu,
\end{equation}
it holds that
\begin{align}
&\pi(\tilde{p}^{(l)}_{j})= U\bar{P}^{(l)}_{j}U^\dagger,\nonumber\\
&\pi(\tilde{x}^{(l)}_{ij})= U\bar{X}^{(l)}_{jk}U^\dagger,\nonumber\\
&\pi(\tilde{y}^{(l)}_{ij})= U\bar{Y}^{(l)}_{jk}U^\dagger,\nonumber\\
\end{align}
for some unitary $U$.

It is then trivial to tackle the $D$-dimensional analogs of (\ref{DNPO_spins}) or (\ref{DNPO_spins_TI}) through hierarchies of SDP relaxations similar to (\ref{Markovian_NPO_rel}) and (\ref{Markovian_NPO_rel_TI}). For finite systems, the higher-dimensional hierarchy will be complete by virtue of Theorem \ref{theorem:Markovian}. For infinite systems, the proof of Theorem \ref{theo:spins} trivially extends to cover higher local dimensions.

When considering thermodynamical systems, we stuck to 1D systems. Dealing with cubic lattices of higher spatial dimension is, however, easy. Take, e.g., a 2D material, in which case each site is labeled by a pair of integers $(j,k)\in\mathbb{Z}^2$, which denote its coordinates in the square lattice. The only changes we need to enforce on the hierarchy (\ref{Markovian_NPO_rel_TI}) to cover this case involve defining \emph{two} translation operators $\t_H,\t_V$, respectively for horizontal and vertical translations, demand the functionals $\tilde{\omega}_{[0,1]}$, $\tilde{b}_0$, $\tilde{b}_1$ to be invariant under the action of both, and replace the sets of monomials $\mathbb{M}_{\alpha}^{1D}$ with $\mathbb{M}_{\alpha}^{2D}(k_\sigma,k_t, N):=\mathbb{M}_\alpha({\cal G}_{2D},k_\sigma,k_t, N)$, with ${\cal G}_{2D}=\{((i,j),(i+1,j))\}\cup\{((i,j),(i,j+1))\}$. The proof of Theorem \ref{theo:spins} does not rely on the 1D geometry, so it immediately generalizes to 2D.

\section{Conclusion}
\label{sec:conclusion}
In this paper, we have introduced the notion of differential non-commutative polynomial optimization problems. We have shown that each such problem can be relaxed into an NPO problem. If the latter happens to be Archimedean, then it is equivalent to the original problem. Since all Archimedean NPO problems admit a complete hierarchy of SDP relaxations, so do all such differential non-commutative polynomial optimization problems.

We tested the speed of convergence of the SDP hierarchy with a number of examples, obtaining reasonably good results with reasonable computational resources. Notably, our hierarchy solves the problem of extrapolating or interpolating a quantum time series, a basic quantum information task for which no general algorithm existed previously. 

Most importantly, our reformulation of DNPO problems into NPOs leads to a complete hierarchy of tractable relaxations for the time-evolved local observables of condensed matter systems, be they finite or infinite. Together with the results in \cite{fawzi2023certified,araujo2024firstorder,steady_state, steady_state2}, they imply that one can model the effect of a quench on ground and Gibbs states of (some other) local Hamiltonians, as well as steady states of local Lindbladians. Our numerical tests suggest that the first levels of the hierarchy already provide very accurate approximations to the local properties of the quenched system, even in the thermodynamic limit. We therefore expect that our algorithm finds application in the study of quenched 2D models, for which the usual approximation methods, e.g., via tensor network states \cite{tns} fail.

\begin{acknowledgments}
\begin{wrapfigure}{r}[0cm]{2cm}
\begin{center}
\includegraphics[width=2cm]{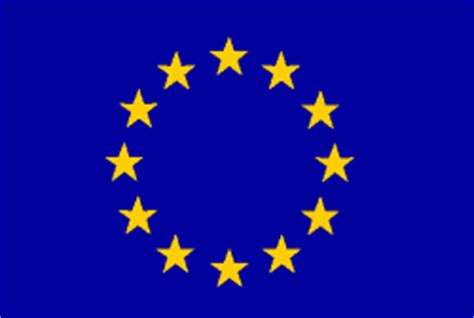}    
\end{center}
\end{wrapfigure} 
We acknowledge useful discussions with Victor Magron, Ilya Kull and Norbert Schuch. This research was funded in whole, or in part, by the Austrian Science Fund (FWF) ([10.55776/I6004] and [10.55776/COE1]) and the European Union– NextGenerationEU. For the purpose of open access, the author has applied a CC BY public copyright license to any Author Accepted Manuscript version arising from this submission. This project was also funded by the QuantERA II Programme, which has received funding from the European Union's Horizon 2020 research and innovation programme under Grant Agreement No 101017733. M.A. acknowledges support by the Spanish Agencia Estatal de Investigación, Grant No. RYC2023-044074-I, by the Q-CAYLE project, funded by the European Union-Next Generation UE/MICIU/Plan de Recuperación, Transformacióon y Resiliencia/Junta de Castilla y León (PRTRC17.11), and also by the Department of Education of the Junta de Castilla y León and FEDER Funds (Reference: CLU-2023-1-05). M.A and A.G. acknowledge funding from the FWF stand-alone project P 35509-N. 
\end{acknowledgments}

\section*{Author contribution statements}
M.N. proposed the SDP hierarchies of relaxations and proved their convergence. M. A. suggested the extra SDP constraint (\ref{constraint_mateus}). All authors participated in the coding of the numerical examples and in the writing of the manuscript.

\bibliographystyle{unsrtnat}
\bibliography{bibliography}

\clearpage

\begin{appendix}
\label{conv_proof}
\section{Proof of Theorem \ref{fund_theo}}
\label{app:proof_theo}
Problem (\ref{problem_relax}) is a relaxation of problem (\ref{problem}), so $q^\star\leq p^\star$. In the next pages we will prove that, if problem (\ref{problem_relax}) is Archimedean, then $p^\star\leq q^\star$, and thus $p^\star=q^\star$. We will prove the claim $p^\star\leq q^\star$ by showing that, for any feasible point of Problem (\ref{problem_relax}), we can construct a feasible point of Problem (\ref{problem}) with the same objective value.

Before proceeding, though, some words on notation. As in the main text, we will use upper case Latin letters to represent operators (e.g.: $\tilde{X}_1(0)$, $X(t)$), while small case letters will denote abstract, non-commuting variables (e.g.: $\tilde{x}_1(0)$, $x(t)$). The symbols $\tilde{\p}_Y,\tilde{\p}_Z$ will respectively denote the set of all polynomials on the variables $\tilde{y}$ or $\tilde{z}:=(\tilde{x},\tilde{y})$. Similarly, $\p$ ($\p_Y$) will denote the set of polynomials on the (also abstract) variables $x(t), y$ ($y$). As opposed to $\tilde{t}$, the symbol $t$ will represent a complex number.

To evaluate abstract variables with actual operators acting on some Hilbert space, we will make use of \emph{operator representations}. Given a set of abstract, non-commuting variables $r=(r_1,\ldots,r_n)$, a representation $\hat{\pi}$ of the set ${\cal R}$ of polynomials in $r_1,\ldots,r_n$ is a map from ${\cal R}$ to $B(\hat{\H})$, for some Hilbert space $\hat{\H}$, with the following properties:
\begin{enumerate}
    \item $\hat{\pi}(1)=\id_{\hat{\H}}$.
    \item For all $p,q\in {\cal R}$, $\hat{\pi}(pq)=\hat{\pi}(p)\hat{\pi}(q)$.
    \item For all $a,b\in\C$, $p,q\in {\cal R}$, $\hat{\pi}(ap+bq)=a\hat{\pi}(p)+b\hat{\pi}(q)$.
    \item For all $p\in{\cal R}$, $\hat{\pi}(p)^\dagger=\hat{\pi}(p^\dagger)$.    
\end{enumerate}
With this terminology, a feasible point of Problem (\ref{problem}) takes the form of a triple $(\H,\psi,\pi)$, respectively representing a Hilbert space, a state for the set $B(\H)$ of bounded operators acting in $\H$ and a map $\pi:\p\to B(\H)$ satisfying the properties above. A solution of Problem (\ref{problem_relax}) similarly corresponds to a triple $(\tilde{\H},\tilde{\omega},\tilde{\pi}:\tilde{\p}\to B(\tilde{H}))$.

We next define a classical system of ordinary differential equations (ODE), related to the differential constraint appearing in Problem (\ref{problem}), which will play an important role in the proof of convergence.
\begin{definition}
For $\gamma>0$, $\tau\in\R$, consider the (classical) system of ordinary differential equations on the vector of real numbers $\xi(t)$:
\begin{align}
\frac{d}{dt}\xi(t)&=g^+(t,\, \xi(t), \gamma\vec{1}^N),\nonumber\\
\xi(\tau)&=\mathbf{\gamma}\vec{1}^M,
\label{maj_ODE}
\end{align}
\noindent where the vector of polynomials $g^+$ is the result of replacing each coefficient of the polynomial $g$ in Problem (\ref{problem}) by its complex modulus and $\vec{1}^K$ denotes the all-ones $K$-dimensional vector. We will call (\ref{maj_ODE}) the \emph{associated ODE} with initial time $\tau$ and magnitude $\gamma$.
\end{definition}

For clarity, we next provide an example of the action of the operator $\bullet^+$:
\begin{example}
For $g_1(X(t))=(1-i)X_1(t)X_2(t)+(1+i)X_2(t)X_1(t)$, we have that $g^+_1(X(t))=\sqrt{2}X_1(t)X_2(t)+\sqrt{2}X_2(t)X_1(t)$.
\end{example}

To prove Theorem \ref{fund_theo}, we will need the following lemma:
\begin{lemma}
\label{time_part_lemma}
Let $\gamma\in\R^+$. Then, there exist $0=t_1<t_2<\ldots<t_{{\cal N}}=1$ such that, for each $k$, the associated ODE with initial time $\tau=t_k$ and magnitude $\gamma$ has an analytic solution $\xi^k(t)$ in a ball $\B(t_k,R_k):=\{z\in\C:|z-t_k|\leq R_k\}$, with $R_k> t_{k+1}-t_k$.
\end{lemma}
\begin{proof}
We begin by defining the \textbf{plus differentiation} linear map $\D^+:\tilde{\p}\to\tilde{\p}$, which satisfies:
\begin{align}
\D^+(\tilde{t}) & = \id{}, \nonumber\\
\D^+(\id) &=\D^+(\tilde{y}_k)=0 \quad\forall k,\nonumber\\
\D^+(\tilde{x}_j)&= g^+_j(\tilde{t},\tilde{x},\tilde{y}) \quad\forall j,\nonumber\\
\D^+(hh')& =\D^+(h)h'+h\D^+(h') \quad \forall h,h'\in\tilde{\p}.
\end{align}
Let $\xi(t;\tau)$ be a solution of (\ref{maj_ODE}), analytic in $t$. Then, for any polynomial $h$, it holds that
\begin{equation}
\frac{d^k h(t,\xi(t;\tau),\gamma\vec{1}^N)}{dt^k}=((\D^{+})^kh)(t,\xi(t;\tau),\gamma\vec{1}^N).
\label{D_plus_t}
\end{equation}
In particular,
\begin{equation}
\frac{d^k h(t,\xi(t;\tau),\gamma\vec{1}^N)}{dt^k}\Bigr|_{t=\tau}=((\D^{+})^kh)(\tau,\gamma\vec{1}^M,\gamma\vec{1}^N),
\label{D_plus}
\end{equation}
where we invoked the boundary condition $\xi(\tau;\tau)=\gamma\vec{1}^M$.

By the Cauchy-Kovalevskaya theorem \cite{book_diff_eqs}, for any $\tau,\gamma\geq 0$, there exists a solution $\xi(t;\tau)$ of the associated ODE with initial time $\tau$ and magnitude $\gamma$, analytic within an open ball of radius $\mu(\tau)>0$ around the point $\tau$. One can thus express $\xi(t;\tau)$ as a Taylor series on $t$ from $t=\tau$, i.e.,
\begin{equation}
\xi(t;\tau)=\sum_{k=0}^\infty\frac{1}{k!}\frac{d^k\xi(t;\tau)}{dt^k}\Bigr|_{t=\tau}(t-\tau)^k=\gamma\vec{1}^M+\sum_{k=1}^\infty\frac{1}{k!}\left((\D^+)^{k-1}g^+\right)(\tau,\gamma\vec{1}^M,\gamma\vec{1}^N)(t-\tau)^k,
\label{expanse}
\end{equation}
where we invoked eq. (\ref{D_plus}). Due to the positivity of the coefficients of $g^+$, each of the Taylor coefficients is a polynomial on $\tau,\gamma$ with real, non-negative coefficients. Moreover, for $0\leq\tau\leq \tau'$ it holds that
\begin{equation}
0\leq \frac{1}{k!}\frac{d^k\xi(t;\tau)}{dt^k}\Bigr|_{t=\tau}=\left((\D^+)^{k-1}g^+\right)(\tau,\gamma\vec{1}^M,\gamma\vec{1}^N)\leq \left((\D^+)^{k-1}g^+\right)(\tau',\gamma\vec{1}^M,\gamma\vec{1}^N)=\frac{1}{k!}\frac{d^k\xi(t;\tau')}{dt^k}\Bigr|_{t=\tau'}.
\end{equation}
That is, the coefficients of the Taylor series for $\xi(t;\tau')$ at $t=\tau'$ dominate those of $\xi(t;\tau)$ at $t=\tau$.
For $\tau\geq 0$, the convergence radius $\mu(\tau)$ is thus a decreasing function. Take then $\mu:=\mu(1)$, which is greater than zero by the Cauchy-Kovalevskaya theorem. Then, for any ${\cal N}\in\N$ with $\mu>\frac{1}{{\cal N}-1}$, the times $\{t_k=\frac{k-1}{{\cal N}-1}\}_{k=1}^{{\cal N}}$ satisfy the conditions of the lemma.
\end{proof}

The rest of the proof of Theorem \ref{fund_theo} is quite complicated; to guide the reader through it, here is a roadmap:
\begin{itemize}
\item Starting from a feasible point $(\tilde{\H},\omega,\tilde{\pi})$ of Problem (\ref{problem_relax}), with $\|\tilde{\pi}(\tilde{t})\|,\|\tilde{\pi}(\tilde{x})\|,\|\tilde{\pi}(\tilde{y})\|\leq\gamma$, we find $t_1,...,t_{\cal N}$ satisfying the conditions of Lemma \ref{time_part_lemma}.

\item Next, we define states $\omega_1,...,\omega_{\cal N}$ with representations $\tilde{\pi}_j:\tilde{\p}\to B(\tilde{\H}_j)$ satisfying $\tilde{\pi}_j(\tilde{t})=t_j$, $\tilde{\pi}_j(\tilde{z}_k)\leq \gamma$, for all $j,k$ and 
\begin{equation}
\omega\left(\Pi_k\tilde{\pi}(\D h)\right)=\omega_{k+1}(\tilde{\pi}_{k+1}(h))-\omega_k(\tilde{\pi}_{k}(h)),
\end{equation}
for all $h\in\tilde{\p}$, where $\Pi_k:=\chi_{[t_{k+1},t_k]}(\tilde{\pi}(\tilde{t}))$. Intuitively, $\tilde{\pi}_j(\tilde{x})$ corresponds to $X(t_j)$ in Problem (\ref{problem}). 

\item For any $j$, the associated ODE with initial time $\tau=t_j$ and magnitude $\gamma$ is analytic in ${\cal B}(t_j, R_j)$, with $R_j>t_{j+1}-t_j$. This implies that the series (\ref{expanse}), with $\tau=t_j$, has convergence radius greater than or equal to $R_j$. Hence, for $j=1,...,{\cal N}-1$, the series
\begin{equation}
\hat{X}^{\underline{j}}(t):=\tilde{\pi}_j(\tilde{x})+\sum_{k=1}^\infty\frac{(t-t_j)^k}{k!}\D^{k-1}g(t_j,\tilde{\pi}_j(\tilde{x}),\tilde{\pi}_j(\tilde{y}))=:s(t,\tilde{\pi}_j(\tilde{x}),\tilde{\pi}_j(\tilde{y});t_j)
\end{equation}
also converges in ${\cal B}(t_j, R_j)$, since the norm of each term in the series is dominated by corresponding term in eq. (\ref{expanse}).

\item It is then proven that
\begin{align}
\omega\left(\Pi_j\tilde{\pi}(h)\right)=\int_{[t_j,t_{j+1}]}dt\omega_j\left(h(t,\hat{X}^{\underline{j}}(t),\tilde{\pi}_j(\tilde{y}))\right),
\label{road_map_ome_conn}
\end{align}
for all polynomials $h\in\tilde{\p}$. Moreover, 
\begin{equation}
\omega_{j}\left(h(\hat{X}^{\underline{j}}(t_{j+1}),\tilde{\pi}_j(\tilde{y}))\right)=\omega_{j+1}\circ\tilde{\pi}_{j+1}(h(\tilde{x},\tilde{y})),
\label{road_plan_omega_k}
\end{equation}
for all $h\in\tilde{\p}_Z$.

\item To construct a feasible point $(\H,\psi,\pi)$ of Problem (\ref{problem}), we set $\H=\tilde{\H}_1$, $\psi=\omega_1$, and define $\pi(y)=\tilde{\pi}_1(\tilde{y})$, $\pi(x(t))=\hat{X}(t)$, with
\begin{align}
&\hat{X}(t)=s(t,\tilde{\pi}_1(\tilde{x}(0)),\tilde{\pi}_1(\tilde{y});t_1),\mbox{ for }t\in [t_1,t_2],\nonumber\\
&\hat{X}(t)=s(t,\hat{X}(t_j),\tilde{\pi}_1(\tilde{y});t_j),\mbox{ for }t\in [t_j,t_{j+1}].
\end{align}
From eq. (\ref{road_plan_omega_k}) and the fact that $\tilde{\pi}_j(\tilde{z}_k)\leq\gamma,\forall j,k$, one shows that the series above converge and that $\hat{X}(1)=\pi(x(1))$. Moreover, by construction $\hat{X}(t)$ satisfies $\frac{d\hat{X}(t)}{dt}=g(t,\hat{X}(t),\pi(y))$.

\item It just remains to prove that $(\H,\psi,\pi)$ satisfies the constraints of Problem (\ref{problem}) and that it has the same objective value as $(\tilde{\H},\omega,\tilde{\pi})$ in Problem (\ref{problem_relax}). The proof relies on (\ref{road_map_ome_conn}), (\ref{road_plan_omega_k}) and the fact that, to certify that some operator $O\in B(\H)$ is positive semidefinite, it suffices to argue that $\psi(\mu(\pi(y)) O\mu(\pi(y))^\dagger)\geq 0$ for all $\mu\in \tilde{\p}_Y$.

\end{itemize}

We continue with the proof. Let $(\tilde{\H},\omega,\tilde{\pi})$ be a feasible point of Problem (\ref{problem_relax}), let $\gamma$ be the maximum norm of the operators $\tilde{\pi}(\tilde{t}),\tilde{\pi}(\tilde{x}),\tilde{\pi}(\tilde{y})$, and let $0=t_1<\ldots<t_{\cal N}=1$ satisfy the conditions of Lemma \ref{time_part_lemma}. We will further assume that $\tilde{\H}$ is separable, that $\omega(\bullet)=\bra{\phi}\bullet\ket{\phi}$, for some $\ket{\phi}\in\tilde{\H}$, and that $\tilde{\H}=\overline{\mbox{span}}\{\tilde{\pi}(a)\ket{\phi}:a\in\tilde{\p}\}$ (this condition is called cyclicity); otherwise, we apply the Gelfand-Neimark-Segal (GNS) construction \cite{takesaki1} to $\omega$ and find a new feasible representation $(\tilde{\H}',\omega',\tilde{\pi}')$ with those two properties and such that $\omega(\tilde{\pi}(p))=\omega'(\tilde{\pi}'(p))$, for all $p\in \p$. 

Notice that the condition $\tilde{\pi}(\tilde{t})-\tilde{\pi}(\tilde{t})^2\geq 0$ implies that the spectrum of $\tilde{\pi}(\tilde{t})$ lies in $[0,1]$. Furthermore, the operator $\tilde{\pi}(\tilde{t})$ has no eigenvalues. Indeed, suppose that $\tau\in [0,1]$ is an eigenvalue of $\tilde{\pi}(\tilde{t})$, and let $\ket{\psi_\tau}\in\tilde{\H}$ be one of its eigenvectors. Then, for any $p\in\tilde{\p}$, we have that 
\begin{equation}
\tilde{\pi}(\tilde{t})\tilde{\pi}(p)\ket{\psi_\tau}=\tilde{\pi}(p)\tilde{\pi}(\tilde{t})\ket{\psi_\tau}=\tau \tilde{\pi}(p)\ket{\psi_\tau}.
\end{equation}
That is, $\tilde{\pi}(p)$ maps eigenvectors of $\tilde{\pi}(\tilde{t})$ to eigenvectors of $\tilde{\pi}(\tilde{t})$ with the same eigenvalue. Now, setting $h=\tilde{t}^k$ in eq. (\ref{integral_conds}), one finds that the distribution $\nu_\omega(t)dt:=\omega(\delta(\tilde{\pi}(\tilde{t})-t))dt$ of values of $\tilde{\pi}(\tilde{t})$ in $\omega$ has the same moments as the uniform distribution in $[0,1]$. Since the spectrum of $\tilde{\pi}(\tilde{t})$ is bounded, both distributions must be the same. This means that $\braket{\phi'_\tau}{\phi}=0$, for any eigenvector $\ket{\phi'_\tau}$ of $\tilde{\pi}(\tilde{t})$ with eigenvalue $\tau$; otherwise, the distribution $\nu_\omega(t)dt$ would have a non-zero mass at $t=\tau$. We therefore have that $\bra{\phi_\tau}\tilde{\pi}(p)\ket{\phi}=0$, for all $p\in\tilde{\p}$, which, for $\ket{\phi_\tau}\not=0$, contradicts the cyclicity of $\ket{\phi}$.

From the above argument, it follows, not only that $\tilde{\pi}(\tilde{t})$ has no eigenvalues, but also that the spectrum of $\tilde{\pi}(\tilde{t})$ is exactly $[0,1]$.

For $k=2,\ldots,{\cal N}-1$, consider the family of polynomials on the real variable $t$ given by
\begin{equation}
d_n(t;t_k):=1-\int_0^t\delta_n(\tau;t_k)d\tau,
\label{def_d_n}
\end{equation}
with
\begin{equation}
\delta_n(t;t_k):=\frac{\Gamma\left(2n+\frac{3}{2}\right)}{\Gamma\left(1+2n\right)\sqrt{\pi}\max(1-t_k,t_k)}\left(1-\left[\frac{t_k-t}{\max(1-t_k,t_k)}\right]^{2}\right)^{2n}.
\label{def_delta_n}
\end{equation}
One can verify that, for $t\in[0,1]$,
\begin{equation}
\lim_{n\to\infty}\delta_n(t;t_k)=\delta(t-t_k), \lim_{n\to\infty}d_n(t;t_k)=\Theta(t_k-t),
\end{equation}
where $\delta(\bullet)$ denotes the Dirac delta distribution and $\Theta(\bullet)$ is the Heaviside function. It follows that
\begin{equation}
\lim_{n\to\infty}\omega\left(d_n(\tilde{\pi}(\tilde{t});t_k)h(\tilde{\pi}(\tilde{t}), \tilde{\pi}(\tilde{z}))\right)=\omega\left(\Theta(t_k-\tilde{\pi}(\tilde{t}))h(\tilde{\pi}(\tilde{t}), \tilde{\pi}(\tilde{z})\right).
\end{equation}
Moreover, since $\tilde{\pi}(\tilde{t})$ has no eigenvalues, $\Theta(t_k-\tilde{\pi}(\tilde{t}))$ is a projector\footnote{Note that, if $\ket{\psi_{t_k}}$ were an eigenvector of $\tilde{\pi}(\tilde{t})$ with eigenvalue $t_k$, then $\Theta(t_k-\tilde{\pi}(\tilde{t}))\ket{\psi_{t_k}}=\frac{1}{2}\ket{\psi_{t_k}}$.}.

On the other hand, for $k=2,\ldots,{\cal N}-1$, the limit of the sequence of functionals
\begin{equation}
\tilde{\omega}_k^{(n)}(h):=\omega\circ \tilde{\pi}\left(\delta_n(\tilde{t};t_k) h\right),
\end{equation}
with $h\in\tilde{\p}$, might not exist. However, for every $h\in\tilde{\p}$, $k=2,\ldots,{\cal N}-1$, the sequence $(\tilde{\omega}_k^{(n)}(h))_n$ is bounded. This follows from the fact that $\tilde{\pi}(\delta_n(\tilde{t};t_k))$ is positive semidefinite (and so $\omega(\delta_n(\tilde{\pi}(\tilde{t});t_k)\bullet)$ is an unnormalized state) and that 
\begin{equation}
\lim_{n\to\infty}\tilde{\omega}_k^{(n)}(1)=\lim_{n\to\infty}\int_{[0,1]}dt\nu_\omega(t)\delta_n(\tilde{t};t_k)=1.
\end{equation}
By the Banach-Alaoglu theorem \cite[Theorem IV.21]{reedsimon}, there exists a subsequence $(n_s)_s$ such that the limits are well-defined, for $k=2,\ldots,{\cal N}-1$ and all $h\in\tilde{\p}$. For $k=2,\ldots,{\cal N}-1$, define thus the linear functionals $\tilde{\omega}_k:\tilde{\p}\to\C$:
\begin{equation}
\tilde{\omega}_k(\bullet):=\lim_{s\to\infty}\omega\circ\tilde{\pi}(\delta_{n_s}(\tilde{t};t_k)\bullet).
\end{equation}
From the previous considerations, it follows that $\tilde{\omega}_k$ is positive and normalized, i.e., $\tilde{\omega}_k(1)=1$. Moreover, for $i=1,\ldots,N+M$ and $\mu\in\tilde{\p}$, it holds that
\begin{align}
&\tilde{\omega}_j(\mu (\gamma^2-\tilde{z}_i^2)\mu^\dagger)=\lim_{s\to\infty}\omega\left(\sqrt{\delta_{n_s}(\tilde{\pi}(\tilde{t});t_j)}\tilde{\pi}(\mu)  (\gamma^2-\tilde{\pi}(z_i)^2)\tilde{\pi}(\mu)^\dagger\sqrt{\delta_{n_s}(\tilde{\pi}(\tilde{t});t_j)}\right)\geq 0.
\label{bound_func}
\end{align}
In addition, for any $\epsilon>0$ it holds that $\lim_{n\to\infty}\delta_n(t;t_k)=0$, for $t\not\in[t_k-\epsilon, t_k+\epsilon]$. Thus, for all $\mu\in \tilde{\p}$,
\begin{align}
&\tilde{\omega}_k(\mu(t_k-\tilde{t}+\epsilon)\mu^\dagger)=\lim_{s\to\infty}\omega\left(\tilde{\pi}(\mu)(t_k-\tilde{\pi}(\tilde{t})+\epsilon)\delta_{n_s}(\tilde{\pi}(\tilde{t});t_k)\tilde{\pi}(\mu)^\dagger\right)\geq 0,\nonumber\\
&\tilde{\omega}_k(\mu(\tilde{t}-t_k+\epsilon)\mu^\dagger)=\lim_{s\to\infty}\omega\left(\tilde{\pi}(\mu)(\tilde{\pi}(\tilde{t})-t_k+\epsilon)\delta_{n_s}(\tilde{\pi}(\tilde{t});t_k)\tilde{\pi}\tilde{\pi}(\mu)^\dagger\right)\geq 0.
\label{single_t_k}
\end{align}
If we apply the GNS construction to $\tilde{\omega}_k$ we will therefore extract a representation $\tilde{\pi}_k:\tilde{\p}\to B(\tilde{\H}_k)$ and a state $\omega_k:B(\tilde{\H}_k)\to\C$ such that
\begin{equation}
\omega_k(\tilde{\pi}_k(h))=\tilde{\omega}_k(h),
\end{equation}
for all $h\in\tilde{\p}$. This representation will satisfy $\|\tilde{\pi}_j(\tilde{z}_i)\|\leq \gamma$ for all $i$  by eq. (\ref{bound_func}) and $-\epsilon\leq\tilde{\pi}_k(\tilde{t})-t_k\leq \epsilon$, for all $\epsilon>0$ by eq. (\ref{single_t_k}). Thus, $\tilde{\pi}_k(\tilde{t})=t_k$.
Altogether, for $k=2,\ldots,{\cal N}-1$, we have that
\begin{align}
&\omega_k(1)=1,\omega_k\geq 0,\omega_k(\tilde{\pi}_k(h))=\tilde{\omega}_k(h),\nonumber\\
&\tilde{\pi}_k(\tilde{t})=t_k,\|\tilde{\pi}_k(\tilde{z}_i)\|\leq \gamma,\forall i.
\label{relations_omega_k}
\end{align}

Finally, define the states $\omega_1=\omega_{\cal N}$ and the representations $\tilde{\pi}_1=\tilde{\pi}_{\cal N}$ by applying the GNS construction to $\omega$, restricted to $\tilde{\pi}(\tilde{\p}_Y)$. We extend the representations $\tilde{\pi}_1,\tilde{\pi}_{\cal N}$ to $\tilde{\p}$ through $\tilde{\pi}_1(\tilde{x}):=\tilde{\pi}_1(\tilde{x}(0))$, $\tilde{\pi}_1(\tilde{t}):=0$, $\tilde{\pi}_{\cal N}(\tilde{x}):=\tilde{\pi}_{\cal N}(\tilde{x}(1))$, $\tilde{\pi}_{\cal N}(\tilde{t})=1$. Since it comes from a GNS construction, $\omega_1$ is cyclic in $\tilde{\pi}_1(\tilde{\p}_Y)$, i.e., for all $A\in B(\tilde{H}_1)$, it holds that
\begin{align}
&\omega_1(\tilde{\pi}_1(\mu) A\tilde{\pi}_1(\nu))=0,\forall\mu,\nu\in\tilde{\p}_Y\Rightarrow A=0,\nonumber\\
&\omega_1(\tilde{\pi}_1(\mu^\dagger) A\tilde{\pi}_1(\mu))\geq 0,\forall\mu\in\tilde{\p}_Y\Rightarrow A\geq 0.
\end{align}

Next, notice that, by the definition of $\D$,
\begin{equation}
\D\left(d_n(\tilde{t};t_k)h(\tilde{t},\tilde{x},\tilde{y})\right)=-\delta_n(\tilde{t};t_k)h(\tilde{t},\tilde{x},\tilde{y})+d_n(\tilde{t};t_k)(\D h)(\tilde{t},\tilde{x},\tilde{y}).
\end{equation}
Setting $n=n_s$, averaging with $\omega$ and taking the limit $s\to\infty$, it follows, by the last line of Problem (\ref{problem_relax}) and eq. (\ref{relations_omega_k}), that
\begin{equation}
\omega\left(\Theta(t_k-\tilde{\pi}(\tilde{t}))\tilde{\pi}\left((\D h)(\tilde{t},\tilde{x},\tilde{y})\right)\right)=\omega_k(\tilde{\pi}_k(h))-\omega_1(\tilde{\pi}_1(h)),
\label{part_int}
\end{equation}
where we invoked the fact that $d_n(1;t_k)\to 0$ when $n\to\infty$. Note that $\Theta(t_k-\tilde{\pi}(\tilde{t}))=\chi_{[0,t_k]}(\tilde{\pi}(\tilde{t}))$, where $\chi_A(\bullet)$ denotes the characteristic function of set $A$. For $k=1,...,{\cal N}-1$, we next define the projectors $\Pi_k:=\chi_{[t_k,t_{k+1}]}(\tilde{\pi}(\tilde{t}))$, which can also be expressed in term of Heaviside functions:
\begin{align}
&\Pi_k=\Theta(t_{k+1}-\tilde{\pi}(\tilde{t}))-\Theta(t_{k}-\tilde{\pi}(\tilde{t})),k=1,\ldots,{\cal N}-1.
\end{align}
From eq. (\ref{part_int}), we have that
\begin{equation}
\omega\left(\Pi_k\tilde{\pi}(\D h)\right)=\omega_{k+1}(\tilde{\pi}_{k+1}(h))-\omega_k(\tilde{\pi}_{k}(h)),
\label{integral_conds_proj}
\end{equation}
for $k=1,\ldots,{\cal N}-2$. The relation above also holds for $k={\cal N}-1$ by virtue of eq. (\ref{integral_conds}). This result is intuitive, if we think the operators $\tilde{\pi}(\tilde{t}),\tilde{\pi}(\tilde{x}),\tilde{\pi}(\tilde{y})$ having the form (\ref{eq:extended_operators}). In that case, eq. (\ref{integral_conds_proj}) follows from the fundamental theorem of calculus.

We next prove that, for $j=1,\ldots,{\cal N}-1$, there exists a trajectory of operators $\{\hat{X}^{\underline{j}}(t):t\in[t_j,t_{j+1}]\}\subset B(\tilde{H}_j)$ that solves the equation
\begin{align}
&\frac{dX(t)}{dt}=g(t,X(t),\tilde{\pi}_j(\tilde{y})),\nonumber\\
&X(t_j)=\tilde{\pi}_j(\tilde{x}).
\label{ODE_step}
\end{align}
and such that
\begin{align}
\omega\left(\Pi_j\tilde{\pi}(h)\right)=\int_{[t_j,t_{j+1}]}dt\omega_j\left(h(t,\hat{X}^{\underline{j}}(t),\tilde{\pi}_j(\tilde{y}))\right),
\label{ODE_step_2}
\end{align}
for all polynomials $h\in\tilde{\p}$.

Let $h\in\tilde{\p}$ be an arbitrary polynomial. By the chain rule $\D(hh')=\D(h)h'+h\D(h')$, we have that
\begin{equation}
\frac{(t_{j+1}-\tilde{t})^n}{n!}\D^n h=-\D\left(\frac{(t_{j+1}-\tilde{t})^{n+1}}{(n+1)!}\D^nh\right)+\frac{(t_{j+1}-\tilde{t})^{n+1}}{(n+1)!}\D^{n+1} h.
\end{equation}
\noindent Applying $\tilde{\pi}$, multiplying by $\Pi_j$ and averaging with $\omega$, we have, by conditions (\ref{integral_conds_proj}), that
\begin{equation}
\omega\left(\Pi_j\frac{(t_{j+1}-\tilde{\pi}(\tilde{t}))^n}{n!}\tilde{\pi}(\D^n h)\right)= \frac{(t_{j+1}-t_j)^{n+1}}{(n+1)!}\omega_j\left(\tilde{\pi}_j(\D^n h)\right) +\omega\left(\Pi_j\frac{(t_{j+1}-\tilde{\pi}(\tilde{t}))^{n+1}}{(n+1)!}\tilde{\pi}(\D^{n+1} h)\right).
\end{equation}
The final term of the right-hand side of the equation above is the same as the left-hand side, with the replacement $n\to n+1$. Thus we can start from $n=0$ and iteratively substitute to arrive at
\begin{align}
&\omega\left(\Pi_j \tilde{\pi}(h)\right)=\sum_{k=0}^n\omega_j\left(\frac{(\D^kh)(t_j,\tilde{\pi}_j(\tilde{x}),\tilde{\pi}_j(\tilde{y}))}{(k+1)!}(t_{j+1}-t_j)^{k+1}\right)+\omega\left(\Pi_j\frac{(t_{j+1}-\tilde{\pi}(\tilde{t}))^{n+1}}{(n+1)!}\tilde{\pi}(\D^{n+1} h)\right)\nonumber\\
&=\int_{[t_j,t_{j+1}]}dt\omega_j\left(\sum_{k=0}^n\frac{(\D^kh)(t_j,\tilde{\pi}_j(\tilde{x}),\tilde{\pi}_j(\tilde{y}))}{k!}(t-t_j)^k\right)+\omega\left(\Pi_j\frac{(t_{j+1}-\tilde{\pi}(\tilde{t}))^{n+1}}{(n+1)!}\tilde{\pi}(\D^{n+1} h)\right).
\label{cosicas}
\end{align}
We next show that the expression evaluated by $\omega_j$ in the second line is the Taylor expansion of $h(t,\hat{X}^{\underline{j}}(t),\tilde{\pi}_j(\tilde{y}))$ from $t=t_j$, with 
\begin{equation}
\hat{X}^{\underline{j}}(t):=s(t,\tilde{\pi}_j(\tilde{x}),\tilde{\pi}_j(\tilde{y});t_j),
\label{piece_sol}
\end{equation}
for
\begin{equation}
s(t,X, Y;t'):=X+\sum_{k=1}^\infty\frac{(t-t')^k}{k!}\D^{k-1}g(t',X,Y).
\label{defin_s}
\end{equation}
To do so, we will need the following lemma.
\begin{lemma}
\label{convergence_lemma}
Let the associated ODE (\ref{maj_ODE}), with initial time $\tau$ and magnitude $\gamma$, admit a solution $\xi(t)$, analytic in $\B(\tau;r)\equiv\{z\in\C:|z-\tau|\leq r\}$. Then, the series $s(t,\bar{X},\bar{Y};\tau)$, as defined in (\ref{defin_s}), converges for $t\in \B(\tau;r)$, $\bar{X}\in\{X:\|X_1\|,\ldots,\|X_N\|\leq \gamma\}=:{\cal X}(\gamma)$, $\bar{Y}\in\{Y:\|\bar{Y}_1\|,\ldots,\|\bar{Y}_M\|\leq \gamma\}=:{\cal Y}(\gamma)$. Moreover, given $\bar{X}\in{\cal X}(\gamma)$, $\bar{Y}\in{\cal Y}(\gamma)$, the series $\bar{X}(t):=s(t,\bar{X},\bar{Y};\tau)$ is, for $t\in [\tau,\tau+r]$, the solution of the system of operator differential equations
\begin{align}
&\frac{d}{dt}\bar{X}(t)=g(t,\bar{X}(t),\bar{Y}),\nonumber\\
\label{diff_lemma}
\end{align}
\noindent with boundary condition $\bar{X}$ at $t=\tau$. 
\end{lemma}
\begin{proof}
To see that the series $s(t,\bar{X},\bar{Y};\tau)$ converges for $t\in{\cal B}(\tau,r)$, $\bar{X}\in{\cal X}(\gamma)$, $\bar{Y}\in{\cal Y}(\gamma)$, it is enough to prove the convergence of
\begin{equation}
\|\bar{X}_i\|+\sum_{k=1}^\infty\frac{1}{k!}\|(\D^{k-1}g_i)(\tau,\bar{X},\bar{Y})\||t-\tau|^k.
\label{norm_series}
\end{equation}

By assumption, for the associated ODE there exists a solution $\xi(t)$, analytic in ${\cal B}(\tau,r)$. Hence, one can express $\xi(t)$ as a Taylor series on $t$ from $t=\tau$ as in eq. (\ref{expanse}), see the proof of Lemma \ref{time_part_lemma}. Each of the Taylor coefficients is a polynomial on $\tau,\gamma$ with real, non-negative coefficients, expressed in terms of the plus differentiation map ${\cal D}^+$. This map has two convenient properties. One is that, for any polynomial $s$ with real, non-negative coefficients, it holds that $(\D^+s)^+=\D^+s$. The other one is that, for any tuple of non-negative scalars/scalar vectors $t,x,y$,
\begin{equation}
(\D^ks)^+(t,x,y)\leq ((\D^+)^ks)(t,x,y).
\label{D_plus_pos}
\end{equation}

Now, for any vector of natural numbers $\vec{i}$, define $Z^{\vec{i}}:=\prod_kZ^{i_k}$. Then, by the properties of the operator norm, the norm of any polynomial $p(Z_1,\ldots,Z_n)=\sum_ip_{\vec{i}}Z^{\vec{i}}$, evaluated on the tuple $Z$, is upper bounded by $\sum_i|p_{\vec{i}}|\prod_{k}\|Z^{i_k}\|=p^+(\|Z\|)$, with $\|Z\|=(\|Z_1\|,\|Z_2\|,\ldots)$. Hence, 
\begin{align}
\left((\D^+)^{k-1}g_i^+\right)(\tau,\gamma\vec{1}^M,\gamma\vec{1}^N)&\geq\left((\D^+)^{k-1}g_i^+\right)(\tau,\|\bar{X}^\tau\|,\|\bar{Y}\|)\nonumber\\
&=\left((\D^+)^{k}\tilde{x}_i\right)(\tau,\|\bar{X}^\tau\|,\|\bar{Y}\|)\nonumber\\
&\overset{(\ref{D_plus_pos})}{\geq} (\D^{k}\tilde{x}_i)^+(\tau,\|\bar{X}^\tau\|,\|\bar{Y}\|),\nonumber\\
&\geq \|(\D^{k}\tilde{x}_i)(\tau,\|\bar{X}^\tau\|,\|\bar{Y}\|)\|\nonumber\\
&= \|(\D^{k-1}g_i)(0,\bar{X}^\tau,\bar{Y})\|.
\end{align}
\noindent for all $k, i$. Since the left-hand side are the Taylor coefficients of $\xi(t)$ and this function is analytic in ${\cal B}(\tau,r)$, it follows that the series (\ref{norm_series}) converges in $t\in{\cal B}(\tau,r)$, and so does the series $s(t,\bar{X},\bar{Y};\tau)$. 

The so-defined function $\bar{X}(t)$ is therefore analytic in $t\in {\cal B}(\tau,r)$, and so is the polynomial vector $g(t,\bar{X}(t),\bar{\X}, \bar{Y})$. The Taylor coefficients of the functions $\frac{d\bar{X}(t)}{dt}$ and $g(t,\bar{X}(t),\bar{\X}, \bar{Y})$ are the same, so $\frac{d\bar{X}(t)}{dt}=g(t,\bar{X}(t),\bar{Y})$. Finally, evaluated at time $t=\tau$, $\bar{X}(t)$ returns the boundary condition $\bar{X}$. We have just proven that $\bar{X}(t)$ is a solution of (\ref{diff_lemma}).
\end{proof}

Now, since the associated ODE with initial time $\tau=t_j$ is analytic in $\B(t_j,t_{j+1}-t_j)$ and the norm of the entries of the operator tuples $\tilde{\pi}_j(\tilde{x}), \tilde{\pi}_j(\tilde{y})$ is bounded by $\gamma$, by Lemma \ref{convergence_lemma} $\hat{X}^{\underline{j}}(t)$ is analytic for $t\in \B(t_j,t_{j+1}-t_j)$. Also, $\hat{X}^{\underline{j}}(t)$ is a solution of the differential equation (\ref{ODE_step}) for $t\in[t_j, t_{j+1}]$. As such, for $t\in [t_j,t_{j+1}]$ it holds that
\begin{equation}
\frac{d^k}{dt^k}h(t,\hat{X}^{\underline{j}}(t),\tilde{\pi}_j(\tilde{y}))=({\cal D}^kh)(t,\hat{X}^{\underline{j}}(t),\tilde{\pi}_j(\tilde{y})).
\end{equation}
Since $\hat{X}^{\underline{j}}(t_j)=\tilde{\pi}_j(\tilde{x})$, we have that
\begin{equation}
\left.\frac{d^k}{dt^k}h(t,\hat{X}^{\underline{j}}(t),\tilde{\pi}_j(\tilde{y}))\right|_{t=t_{j}}=({\cal D}^kh)(t_j,\tilde{\pi}_j(\tilde{x}),\tilde{\pi}_j(\tilde{y})).
\end{equation}
Thus, as promised, the Taylor expansion of $h(t,\hat{X}^{\underline{j}}(t),\tilde{\pi}_j(\tilde{y}))$ from $t=t_j$ coincides with the operator evaluated by $\omega_j$ on the second line of (\ref{cosicas}). Since $h$ is a polynomial, the former function is also analytic in $t\in\B(t_j,t_{j+1}-t_j)$ and hence, for $t\in [t_j,t_{j+1}]$, the series converges to $h(t,\hat{X}^{\underline{j}}(t),\tilde{\pi}_j(\tilde{y}))$ in the limit $n\to\infty$.

In addition, the norm of the operator in the second term of the right-hand side of eq. (\ref{cosicas}) tends to zero for $n\to\infty$. This is so because the norm of the evaluated operator is upper bounded by
\begin{equation}
\max_{t\in[t_j,t_{j+1}]}\left\|\frac{(t_{j+1}-t)^{n+1}}{(n+1)!}\tilde{\pi}(\D^{n+1} h)(t,\tilde{x},\tilde{y})\right\|,
\label{normi_bound}
\end{equation}
which follows from the spectral measure of $\tilde{\pi}(\tilde{t})$, the fact that $\tilde{\pi}(\tilde{t})$ commutes with $\tilde{\pi}(\tilde{\p})$ and the definition of $\Pi_j$. Let $t^\star$ be the maximizer of the expression above. Then, the argument of the operator norm $\|\bullet\|$ in (\ref{normi_bound}) is the $n^{th}$ coefficient of the Taylor expansion of $h(t,\check{X}(t),\tilde{\pi}(\tilde{y}))$, from $t=t^\star$ to $t=t_{j+1}$, with
\begin{equation}
\check{X}(t):=s(t,\tilde{\pi}_j(\tilde{x}),\tilde{\pi}(\tilde{y});t^\star).
\end{equation}
We claim that the fact that the associated ODE with initial time $\tau=t_j$ and magnitude $\gamma$ is analytic in $\B(t_j,t_{j+1}-t_j)$ implies that the associated ODE with initial time $\tau=t^\star$ and magnitude $\gamma$ is analytic in $\B(t^\star,t_{j+1}-t^\star)$. Therefore, by Lemma \ref{convergence_lemma}, $\check{X}(t)$ is analytic in $\B(t^\star,t_{j+1}-t^\star)$, and so is $h(t,\check{X}(t),\tilde{\pi}(\tilde{y}))$. Consequently, its Taylor term tends to zero for $n\to\infty$, and so does the last term of eq. (\ref{cosicas}).

To prove the claim, first notice that the solution $\xi(t;t_j)$ of the associated ODE with initial time $\tau=t_j$ and magnitude $\gamma$ is also analytic in $\B(t^\star,t_{j+1}-t^\star)$, as this region is contained in $\B(t_j,t_{j+1}-t_j)$. Thus, one can Taylor expand $\xi(t;t_j)$ from $t^\star$ to $t_{j+1}$, and, due to the form of eq. (\ref{maj_ODE}), its $n^{th}$ Taylor coefficient can be expressed as $q_n(t^\star,\xi(t^\star;t_j))$ for some polynomial $q_n$ with positive coefficients. On the other hand, the $n^{th}$ Taylor coefficient of the solution $\xi(t;t^\star)$ of the associated ODE with initial time $\tau=t^\star$ and magnitude $\gamma$ is $q_n(t^\star,\gamma\vec{1}^M)$. Since $\xi(t;t_j)$ increases with $t$, we thus have that $\xi(t^\star;t_j)\geq \gamma\vec{1}^M$ entry-wise, and therefore $q_n(t^\star,\xi(t^\star;t_j))\geq q_n(t^\star,\gamma\vec{1}^M)$. The Taylor coefficients of $\xi(t;t_j)$ majorize those of $\xi(t;t^\star)$, and so the Taylor series of the latter must also converge for $t\in \B(t_j,t_{j+1}-t_j)$. The claim is proven.

We conclude that, for any polynomial $h\in\tilde{\p}$,
\begin{equation}
\omega\left(\Pi_j\tilde{\pi}(h)\right)=\int_{[t_j,t_{j+1}]}dt\omega_j\left(\sum_{k=0}^\infty\frac{\tilde{\pi}_j(\D^nh)}{k!}(t-t_j)^k\right)=\int_{[t_j,t_{j+1}]}dt\omega_j\left(h(t,\hat{X}^{\underline{j}}(t),\tilde{\pi}_j(Y))\right),
\label{taylor}
\end{equation}
hence proving eq. (\ref{ODE_step}).

Let $h\in\tilde{\p}_Z$. Then,
\begin{align}
&\omega_{j}\left(h(\hat{X}^{\underline{j}}(t_{j+1}),\tilde{\pi}_j(\tilde{y}))\right)-\omega_j\left(\tilde{\pi}_j(h)\right)\nonumber\\
&=\omega_{j}\left(h(\hat{X}^{\underline{j}}(t_{j+1}),\tilde{\pi}_j(\tilde{y}))\right)-\omega_{j}\left(h(\hat{X}^{\underline{j}}(t_{j}),\tilde{\pi}_j(\tilde{y}))\right)\nonumber\\
&=\int_{[t_j,t_{j+1}]}dt\omega_j\left(\frac{dh(\hat{X}^{\underline{j}}(t),\tilde{\pi}_j(Y))}{dt}\right)\nonumber\\
&=\int_{[t_j,t_{j+1}]}dt\omega_j\left(\D(h)(\hat{X}^{\underline{j}}(t),\tilde{\pi}_j(Y))\right)\nonumber\\
&\overset{(\ref{taylor})}{=}\omega\left(\Pi_j\D(h)(\tilde{x},\tilde{y})\right)\nonumber\\
&\overset{(\ref{integral_conds_proj})}{=}\omega_{j+1}\left(\tilde{\pi}_{j+1}(h)\right)-\omega_j\left(\tilde{\pi}_j(h)\right).
\end{align}
It follows that, for all $h\in\tilde{\p}_Z$,
\begin{equation}
\omega_{j}\left(h(\hat{X}^{\underline{j}}(t_{j+1}),\tilde{\pi}_j(\tilde{y}))\right)=\omega_{j+1}\circ\tilde{\pi}_{j+1}(h(\tilde{x},\tilde{y})).
\label{relation_states}
\end{equation}
For arbitrary $j=1,\ldots,{\cal N}-1$ and $k\leq j$, we define the function $s^j(\bar{Y})$ by recursion:
\begin{align}
&s^1(\bar{Y})=\bar{X}(0),\nonumber\\
&s^{j+1}(\bar{Y})=s(t_{j+1},s^{j}(\bar{Y}),\bar{Y};t_j).
\end{align}
That way, we can easily formulate the next proposition.
\begin{prop}
\label{prop_relations}
For $j=1,\ldots,{\cal N}$, the series $s^{j}(\tilde{\pi}_1(\tilde{y}))$ converges and its entries have norm bounded by $\gamma$. Moreover, for any polynomial $\mu\in\tilde{\p}_Z$,
\begin{equation}
\omega_1(\mu(s^{j}(\tilde{\pi}_1(\tilde{y})),\tilde{\pi}_1(\tilde{y})))=\omega_j(\mu(\tilde{\pi}_j(\tilde{x}),\tilde{\pi}_j(\tilde{y}))).
\label{rel_1_j}
\end{equation}
In addition, $s^{{\cal N}}(\tilde{\pi}_1(\tilde{y}))=\tilde{\pi}_1(\tilde{x}(1))$.
\end{prop}
\begin{proof}
We will prove the first part of the proposition (namely, eq. (\ref{rel_1_j}) and $\|s^{j}(\tilde{\pi}_1(\tilde{y}))_i\|\leq\gamma$, $i=1,\ldots,M$) by induction on $j$. For $j=1$, the claim is obvious, so let us assume that the claim holds for $j-1$. First, since $s^{j-1}(\tilde{\pi}_1(\tilde{y}))$ is bounded by $\gamma$, by Lemma \ref{convergence_lemma} it follows that the series 
\begin{equation}
s^{j}(\tilde{\pi}_1(\tilde{y}))=s(t_{j},s^{j-1}(\tilde{\pi}_1(\tilde{y})),\tilde{\pi}_1(\tilde{y});t_{j-1})
\end{equation}
converges. Call $s(t,\bar{X},\bar{Y};t',n)$ the polynomial in $\bar{X},\bar{Y}$ that results when we sum the first $n$ terms of the series (\ref{defin_s}). By the induction hypothesis, we have that, for any polynomial $\mu\in\tilde{\p}_Z$ and any $n\in\N$,
\begin{equation}
\omega_1\left(\mu(s(t_j,s^{j-1}(\tilde{\pi}_1(\tilde{y})),\tilde{\pi}_1(\tilde{y});t_{j-1},n),\tilde{\pi}_1(\tilde{y})\right)=\omega_{j-1}\left(\mu(s(t_j,\tilde{\pi}_{j-1}(\tilde{x}),\tilde{\pi}_{j-1}(\tilde{y});t_{j-1},n),\tilde{\pi}_{j-1}(\tilde{y})\right).
\end{equation}
Since the series on both sides converge, we can take the limit $n\to\infty$. This results in the identity
\begin{equation}
\omega_1\left(\mu(s^j(\tilde{\pi}_1(\tilde{y})),\tilde{\pi}_1(\tilde{y}))\right)=\omega_{j-1}\left(\mu(\hat{X}^{\underline{j-1}}(t_{j}),\tilde{\pi}_{j-1}(\tilde{y}))\right).
\end{equation}
By eq. (\ref{relation_states}), the right-hand side equals the right-hand side of eq. (\ref{rel_1_j}). Thus, eq. (\ref{rel_1_j}) is proven to hold for $j$ as well.

To prove that the entries of $s^j(\tilde{\pi}_1(\tilde{y}))$ are bounded by $\gamma$, we set
\begin{equation}
\mu(\tilde{x},\tilde{y}):= \nu(\tilde{y})\left(\gamma^2-\tilde{x}^2_i\right)\nu(\tilde{y})^\dagger
\end{equation}
in eq. (\ref{rel_1_j}). This results in the relation
\begin{equation}
\omega_1\left(\nu(\tilde{\pi}_1(\tilde{y}))\left(\gamma^2-s^j_i(\tilde{\pi}_1(\tilde{y}))^2\right)\nu(\tilde{\pi}_1(\tilde{y}))^\dagger\right)=\omega_j\left(\nu(\tilde{\pi}_j(\tilde{y}))\left(\gamma^2-(\tilde{\pi}_j(\tilde{x}_i))^2\right)\nu(\tilde{\pi}_j(\tilde{y}))^\dagger\right)\geq 0,
\end{equation}
where the inequality follows from the fact that $\tilde{\pi}_j(\tilde{x}_i)$ is bounded by $\gamma$. Since the relation holds for arbitrary $\nu$ and $\omega_1$ is cyclic, we have that the operator between the $\nu$'s on the left-hand side of the equation above is positive semidefinite. Thus, $\|s^j(\tilde{\pi}_1(\tilde{y}))\|\leq\gamma$. This completes the induction step and so the first part of the lemma holds for all $j$.

To finish the proof, we need to show that $s^{{\cal N}}(\tilde{\pi}_1(\tilde{y}))=\tilde{\pi}_1(\tilde{x}(1))$. This follows from eq. (\ref{rel_1_j}) with $j={\cal N}$. Indeed, set $\mu=\nu(\tilde{y})(\tilde{x}_i- \tilde{x}(1))\nu'(\tilde{y})$. Then we have that
\begin{equation}
\omega_1\left(\nu(\tilde{\pi}_1(\tilde{y}))(s^{\cal N}(\tilde{\pi}_1(\tilde{y}))- \tilde{\pi}_1(\tilde{x}(1)))\nu'(\tilde{\pi}_1(\tilde{y}))\right)=\omega_{\cal N}\left(\nu(\tilde{\pi}_{\cal N}(\tilde{y}))(\tilde{\pi}_{\cal N}(\tilde{x}_i)- \tilde{\pi}_{\cal N}(\tilde{x}(1)))\nu'(\tilde{\pi}_{\cal N}(\tilde{y}))\right)=0,
\end{equation}
where the last equality follows from the identity $\tilde{\pi}_{\cal N}(\tilde{x}_i)=\tilde{\pi}_{\cal N}(\tilde{x}(1))$. As before, since the relation holds for all polynomials $\nu,\nu'$ and $\omega_1$ is cyclic, it follows that $s^{{\cal N}}(\tilde{\pi}_1(\tilde{y}))=\tilde{\pi}_1(\tilde{x}(1))$.

\end{proof}
Proposition \ref{prop_relations} allows us to replace expressions containing of $\omega_j$, $j=1,\ldots,{\cal N}$ by the corresponding expressions with $\omega_1$. For ease of notation, from now on we rename $\tilde{\H}_1\to\H$ and $\omega_1\to\psi$, and define the representation $\pi:\p_Y\to B(\H)$ through $\pi(p(y)):=\tilde{\pi}_1(p(\tilde{y}))$. Of course, $\psi$ inherits from $\omega_1$ its cyclicity with respect to $\pi(\p_Y)$.

Now, consider the trajectory of operators in $B(\H)$ given by:
\begin{equation}
\hat{X}(t):=s(t,s^{j}(\pi(y)),\pi(y)),\mbox{ for } t_j\leq t\leq t_{j+1}.
\end{equation}
This is by construction a solution of the differential equation
\begin{equation}
\frac{d\hat{X}(t)}{dt}=g(t,\hat{X}(t),\pi(y)),
\end{equation}
for $t\in[0,1]$, with boundary conditions $\hat{X}(0)=\pi(x(0))$, $\hat{X}(1)=\pi(x(1))$. From now on, we thus extend the representation $\pi$ from $\p_Y$ to $\p$, through $\pi(x(t)):=\hat{X}(t)$.

In the next lines, we prove that $(\H,\psi,\pi)$ is a feasible point of Problem (\ref{problem}), with the same objective value as $(\tilde{\H},\omega,\tilde{\pi})$.

We first tackle the objective value. First, note that, in the $\psi,\pi$ notation, Proposition \ref{prop_relations} implies, for $t\in [t_j,t_{j+1}]$, that
\begin{equation}
\psi(h(t,s(t,s^j(\pi(y)),\pi(y);t_j,n),\pi(y))=\omega_j(h(t,s(t,\tilde{\pi}_j(\tilde{x}),\tilde{\pi}_j(\tilde{y});t_j,n),\tilde{\pi}_j(\tilde{y}))).
\label{ident_omegas_n}
\end{equation}
Taking the limit $n\to\infty$, we have
\begin{equation}
\psi(h(t,\pi(x(t)),\pi(y)))=\omega_j(h(t,X^{\underline{j}}(t),\tilde{\pi}_j(\tilde{y}))), \mbox{ for }t\in [t_j,t_{j+1}].
\label{ident_omegas}
\end{equation}

Now, let $h\in\tilde{\p}$. We have that
\begin{align}
&\int_{[0,1]}dt\psi(h(t,\pi(x(t)),\pi(y)))\nonumber\\
&=\sum_{j=1}^{{\cal N}-1}\int_{[t_j,t_{j+1}]}dt\psi(h(t,\pi(x(t)),\pi(y)))\nonumber\\
&\overset{(\ref{ident_omegas})}=\sum_{j=1}^{{\cal N}-1}\int_{[t_j,t_{j+1}]}dt\omega_j(h(t,X^{\underline{j}}(t),\tilde{\pi}_j(\tilde{y})))\nonumber\\
&\overset{(\ref{taylor})}{=}\sum_{j=1}^{{\cal N}-1}\omega(\Pi_j\tilde{\pi}\circ h(\tilde{t},\tilde{x},\tilde{y}))=\omega(\tilde{\pi}\circ h(\tilde{t},\tilde{x},\tilde{y})),
\label{connection_psi_omega}
\end{align}
where we made use of the identity $\sum_j\Pi_j=1$.

Setting $h=f$ in eq. (\ref{connection_psi_omega}), we have that
\begin{equation}
\psi\left(\int_{[0,1]} dt f(t,\pi(x(t)),\pi(y))\right)=\omega\left( \tilde{\pi}\circ f(\tilde{t},\tilde{x},\tilde{y})\right),
\end{equation}
\noindent i.e., the objective values of $(\H,\psi,\pi)$, $(\tilde{\H},\omega,\tilde{\pi})$ are the same.

We still need to prove that $(\H,\psi,\pi)$ satisfies the constraints of Problem (\ref{problem}). In this regard, the state-dependent constraints also follow straightforwardly from eq. (\ref{connection_psi_omega}):
\begin{equation}
\int_{[0,1]}dt\psi(\pi(q_j(t,x(t),y)))=\omega\circ \tilde{\pi}(q_j(\tilde{t},\tilde{x},\tilde{y}))\geq 0.
\end{equation}

It remains to prove that $\pi\circ p_k(t,\pi(x(t)),\pi(y))$ is a positive semidefinite operator for $t\in[0,1]$ and $k=1,\ldots,K$. Since $\psi$ is cyclic in $\pi(\p_Y)$, this is equivalent to showing that
\begin{equation}
\psi\left(h^\dagger(\pi(y))\left(\int_{[0,1]} dt \nu(t) p_k(t,\pi(x(t)), \pi(y))\right)h(\pi(y))\right)\geq 0,
\end{equation}
\noindent for all non-negative polynomials $\nu(t)$ in $[0,1]$ and all polynomials $h$. By eq. (\ref{connection_psi_omega}), the left-hand side equals
\begin{equation}
\omega\left( h^\dagger(\tilde{\pi}(\tilde{y}))\nu(\tilde{\pi}(\tilde{t})) p_k(\tilde{\pi}(\tilde{t}),\tilde{\pi}(\tilde{x}),\tilde{\pi}(\tilde{y}))h(\tilde{\pi}(\tilde{y}))\right).
\end{equation}
\noindent The averaged operator above is positive semidefinite\footnote{This follows from the fact that $\tilde{\pi}\circ \nu(\tilde{t}), \tilde{\pi}\circ p_k(\tilde{t},\tilde{x},\tilde{y})$ are positive semidefinite and $\tilde{\pi}\circ \nu(\tilde{t})$ commutes with $\tilde{\pi}\circ p_k(\tilde{t},\tilde{x},\tilde{y})$.}, and so its average over state $\omega$ is non-negative. 

In sum, $(\H,\psi,\pi)$ is a feasible point of Problem (\ref{problem}), with the same objective value as $(\tilde{\H},\omega,\tilde{\pi})$. Theorem \ref{fund_theo} has been proven.

\section{Extensions of the hierarchy}
\label{app:extensions}
In this appendix, we show how to deal with DNPO problems with more exotic constraints.

\subsection{Time-dependent average constraints}
Suppose that we are interested in solving a variant of problem $\mathbf{P}$ (\cref{problem}), with state constraints of the form (\ref{eq:time_const}). Namely, we wish to enforce time-dependent restrictions on the quantum state. 

As it turns out, such constraints admit an easy representation in terms of the auxiliary operators $\tilde{X},\tilde{Y}, \tilde{T}$ and the state $\omega$. 

First of all, under eq. (\ref{aver_corresp}), the relation \eqref{eq:time_const} is equivalent to the condition:
\begin{equation}
\omega\left(P(\tilde{T})Q_j(\tilde{T},\tilde{X}, \tilde{Y})\right)\geq 0\quad\mbox{for }j=1,\ldots,J,\mbox{ for all polynomials } P \mbox{ with }P(t)\geq 0,\forall t\in [0,1].
\label{time_const2}
\end{equation}
Indeed, suppose that eq. (\ref{eq:time_const}) holds and let $P(t)$ be a polynomial with $P(t)\geq 0$, for $t\in [0,1]$. By eq. (\ref{aver_corresp}) we have that
\begin{equation}
\omega\left(P(\tilde{T})Q_j(\tilde{T},\tilde{X}, \tilde{Y})\right)=\int_{[0,1]} P(t)dt\psi\left(Q_j(t,X(t), Y)\right)\geq 0,
\end{equation}
Let us next assume that (\ref{time_const2}) holds. For any $\hat{t}\in [0,1]$, there exists a sequence of polynomials $(P_n(t))_n$, positive in $t\in [0,1]$, such that
\begin{equation}
\lim_{n\to\infty} P_n(t)=\delta(t-\hat{t}).
\end{equation}
Invoking eq. (\ref{aver_corresp}), we have that
\begin{equation}
0\leq \lim_{n\to\infty}\omega\left(P_n(\tilde{T})Q_j(\tilde{T},\tilde{X}, \tilde{Y})\right)=\psi\left(Q_j(\hat{t},X(\hat{t}), Y)\right),
\end{equation}
and so eq. (\ref{eq:time_const}) follows.

Now, by \cite{Fekete}, a polynomial $P(t)$ on a variable $t\in[0,1]$ is positive if and only if it admits a decomposition into polynomials $\{f_j\}_j$, $\{f_j^+\}_j$ and $\{f_j^-\}_j$ of the form
\begin{equation}
P(t)=\sum_j f_j(t)^2+\sum_j tf^+_j\!(t)^2+\sum_j (1-t)f^{-}_j(t)^2.    
\end{equation}
\noindent Hence, condition \eqref{time_const2} can be enforced in the SDP relaxations of problem $\mathbf{Q}$ by imposing the positive semi-definiteness of the matrices $\{\Gamma^j, \Gamma^{+,j},\Gamma^{-,j}\}_{j=1}^J$, defined through the relations:
\begin{align}
\Gamma^j_{kl}&=\tilde{\omega}(\tilde{t}^{k+l}Q_j(\tilde{t},\tilde{x}, \tilde{y})) \nonumber\\
\Gamma^{j,+}_{kl}&=\tilde{\omega}(\tilde{t}^{k+l+1}Q_j(\tilde{t},\tilde{x}, \tilde{y})), \nonumber\\
\Gamma^{j,-}_{kl}&=\tilde{\omega}(\tilde{t}^{k+l}(1-\tilde{t})Q_j(\tilde{t},\tilde{x},\tilde{y})).
\end{align}
The proof of Theorem \ref{fund_theo} generalizes when constraints of the form (\ref{eq:time_const}) are dealt with in this way. Thus, under the Archimedean condition, the corresponding hierarchy of SDPs converges to the solution of the differential problem.

\subsection{More than two boundary conditions}
\label{sec:multiple_boundaries}
Another interesting extension of problem $\mathbf{P}$ is a scenario where the vector variable $Y$ includes, not just the boundary values $X(0)$, $X(1)$, but $(X(t_1), \ldots,X(t_N))$, for $\{t_k\}_k\subset [0,1]$, with $t_1=0, t_N=1$. 
In order to relax this problem to a non-commutative polynomial optimization problem, we define the extended operators $\tilde{X}, \tilde{Y}, \tilde{T}$ as before, and add projector-type variables $\Pi_k\equiv\Theta(t_{k+1}-\tilde{T})-\Theta(t_{k}-\tilde{T})$, where $\Theta$ denotes the Heaviside function, to the lot. These projectors not only commute with the auxiliary variables $\tilde{X},\tilde{Y}, \tilde{T}$, but also satisfy $\sum_k\Pi_k=\id$ and, most importantly, the relation
\begin{equation}
\omega\left(\D(h)(\tilde{T},\tilde{X},\tilde{Y})\Pi_k\right)=\omega\left(h(t_{k+1},\tilde{X}(t_{k+1}),\tilde{Y})\right)-\omega\left(h(t_{k},\tilde{X}(t_k), \tilde{Y})\right).
\end{equation}
The proof of Theorem \ref{fund_theo} similarly extends to this scenario.

\section{Proof of Theorem \ref{theorem:Markovian}}
\label{app:Markovian}
The proof of Theorem \ref{theorem:Markovian} is very similar to that or Theorem \ref{fund_theo}. There is, however, an extra difficulty. This time, the state $\omega_1$ is cyclic with respect to the algebra generated by the operators $\tilde{\pi}_1(\tilde{y}_{\text{rest}}), \tilde{\pi}_1(\tilde{x})$. Namely, to prove that some operator $O\in B(\tilde{H}_1)$ is positive semidefinite, one must show that 
\begin{equation}
\omega_1(\mu(\tilde{\pi}_1(\tilde{x}),\tilde{\pi}_1(\tilde{y}_{\text{rest}})) O \mu(\tilde{\pi}_1(\tilde{x}),\tilde{\pi}_1(\tilde{y}_{\text{rest}}))^\dagger)\geq 0, \forall \mu\in \tilde{\p}_0.
\end{equation}
Nonetheless, due to the structure of Problem (\ref{problem_relax_relaxed}), very often we will rather derive something like this:
\begin{equation}
\omega_1(\mu(\hat{X}(t),\tilde{\pi}_1(\tilde{y}_{\text{rest}})) O \mu(\hat{X}(t),\tilde{\pi}_1(\tilde{y}_{\text{rest}}))^\dagger)\geq 0,\forall \mu\in \tilde{\p}_0.
\end{equation}

Our first step is to ensure that the two conditions are equivalent. This is tackled by the next lemma.


\begin{lemma}
\label{lemma:C_algebra_eq}
Let $X(t)=(X_1(t),\ldots,X_M(t))$ be a solution of the differential equation:
\begin{align}
&\frac{dX}{dt}=g(t,X(t),Y),
\end{align}
for $t\in[t_0,t_f]$ and, for each $i$ and some $R>0$, let $X_i(t)$ be analytic within the region ${\cal R}:=\{t:|\mbox{Im}(t)|\leq R,\mbox{Re}(t)\in [-R+t_0,t_f+R]\}$. Then, for any $t\in [t_0,t_f]$, the $C^*$-algebra $\mathbb{A}(t)$ generated by $X(t),Y$ does not depend on $t$.
\end{lemma}
\begin{proof}
Let $\tau,\tau'\in[t_0,t_f)$, with $r:=|\tau-\tau'|< R$. Then $X_i(t)$ is analytic within the ball $\{|t-\tau|\leq r\}\subset {\cal R}$, which implies that one can expand $X_i(t)$ as a series of powers in $t-\tau$ all the way up to $t=\tau'$. Thus,
\begin{equation}
X_i(\tau')=X_i(\tau)+\sum_{k=1}^\infty \frac{\D^{k-1}(g)(t,X(t),Y)}{k!}\Bigr|_{t=\tau}(\tau'-\tau)^k.
\end{equation}
This implies that $\{X_i(\tau')\}_i\subset\mathbb{A}(\tau)$. Thus, $\mathbb{A}(\tau')\subset \mathbb{A}(\tau)$. The argument is symmetric and so it also holds that $\mathbb{A}(\tau)\subset \mathbb{A}(\tau')$, and therefore $\mathbb{A}(\tau)=\mathbb{A}(\tau')$. Now, for any $\tau,\tau'\in [t_0,t_f]$, with $\tau<\tau'$, one can find $R>\Delta>0, n\in\N$ such that $\tau+\Delta n=\tau'$ and the ball $\{|t-\tau+\Delta k|\leq \Delta\}$ is contained in $\subset {\cal R}$, for $k=0,\ldots,n$. By the previous argument, we therefore have that $\mathbb{A}(\tau)=\mathbb{A}\left(\tau+\Delta\right)=\ldots=\mathbb{A}(\tau')$.
\end{proof}

\begin{proof}[Proof of Theorem \ref{theorem:Markovian}]
The last statement of the theorem follows from the proof of \cite[Theorem $1$]{NPO}: in a nutshell, if the Archimedean condition holds independently for the three sets of variables, then the solution of each SDP relaxation can be taken bounded, and the Banach-Alaoglu theorem \cite[Theorem IV.21]{reedsimon} guarantees that there is a converging subsequence of the resulting sequence of triples of partially defined functionals. The limiting triple will satisfy all constraints of problem (\ref{problem_relax_relaxed}), and thus the sequence of SDP lower bounds converges to the solution of (\ref{problem_relax_relaxed}).

For the equivalence between (\ref{problem_relax2}) and (\ref{problem_relax_relaxed}), the second one is clearly a relaxation of the first one. Thus, to prove the theorem it suffices to show that any tuple of feasible functionals $\tilde{\omega}_{[0,1]},\tilde{b}_{0}, \tilde{b}_{1}$ satisfying the constraints of Problem (\ref{problem_relax_relaxed}) can be used to construct feasible operators $X(t),Y_{\text{rest}}$ and a state $\psi$ such that
\begin{align}
&p^{0}_k(X(0),Y_{\text{rest}})\geq 0,\mbox{ for }k=1,\ldots,K_{0},\nonumber\\
&p^{1}_k(X(0),Y_{\text{rest}})\geq 0,\mbox{ for }k=1,\ldots,K_{1},\nonumber\\
&p^{[0,1]}_k(t,X(t),Y_{\text{rest}})\geq 0,\mbox{ for }k=1,\ldots,K_{[0,1]},\nonumber\\
&\frac{dX(t)}{dt}=g(t,X(t), Y_{\text{rest}}),\nonumber\\
&\int_{[0,1]}dt\psi(h(t,X(t),Y_{\text{rest}})=\tilde{\omega}_{[0,1]}(h(\tilde{t},\tilde{x},\tilde{y}_{\text{rest}})),\forall h\in \p_{[0,1]},\nonumber\\
&\psi(h(X(0),Y_{\text{rest}})=\tilde{b}_{0}(h(\tilde{x}(0),\tilde{y}_{\text{rest}})),\forall h\in \p_0,\nonumber\\
&\psi(h(X(1),Y_{\text{rest}})=\tilde{b}_{1}(h(\tilde{x}(1),\tilde{y}_{\text{rest}})),\forall h\in \p_0.
\label{conds_theor_Markovian}
\end{align}
The remaining state constraints of problem (\ref{problem2}), as well as the preservation of the value of the objective function, follow from the last three conditions.

As already anticipated, to demonstrate the above we will closely follow the proof of Theorem \ref{fund_theo}. Given the functionals $\tilde{\omega}_{[0,1]},\tilde{b}_{0}, \tilde{b}_{1}$, let $(\tilde{\H}_{[0,1]},\omega,\tilde{\pi}_{[0,1]})$, $(\tilde{\H}_{0},b_0,\bar{\pi}_{0})$, $(\tilde{\H}_{1},b_1,\bar{\pi}_{1})$ respectively be the result of applying the GNS construction to each of them. We can assume (see the proof of Theorem \ref{fund_theo} for details) that the three Hilbert spaces are separable and that $\tilde{\pi}_{[0,1]}(\tilde{t})$ has spectrum $[0,1]$ and no eigenvalues. Define $\gamma$ to be the greatest norm of $\{\tilde{\pi}_{\alpha}(\bar{z}_i):i\}_{\alpha=0,1}\cup\{\tilde{\pi}_{[0,1]}(\tilde{z}_i)\}$ and let $0=t_1<\ldots<t_{\cal N}=1$ be such that, for each $k$, the associated ODE (\ref{maj_ODE}) at $t_k$ and initial value $\gamma$ is analytic in a ball $|t-t_k|\leq R_k$, with $R_k>|t_{k+1}-t_k|$ (the existence of such a finite set $\{t_k\}_k$ is guaranteed by Lemma \ref{time_part_lemma}).

To adapt to the notation of the proof of Theorem \ref{fund_theo}, we define the states $\omega_1:=b_0,\omega_{\cal N}:=b_1$; and the representations $\tilde{\pi}_1=\bar{\pi}_0$, $\tilde{\pi}_{\cal N}=\bar{\pi}_1$. We also define the states $\{\omega_k\}_{k=2}^{{\cal N}-1}$ and representations $\{\tilde{\pi}_k\}_{k=2}^{{\cal N}-1}$, as the result of applying the GNS construction to any limiting subsequence of the functionals
\begin{equation}
\tilde{\omega}^{(n)}_k(\bullet):=\omega(\delta_n(\tilde{\pi}_{[0,1]}(\tilde{t});t_k)\tilde{\pi}_{[0,1]}(\bullet)),
\end{equation}
with $\delta_n(t;t_k)$ defined as in (\ref{def_delta_n}). As shown in the proof of Theorem \ref{fund_theo}, $\{\omega_k\}_{k=2}^{{\cal N}-1}$ are positive, normalized states, $\|\tilde{\pi}_k(\tilde{z}_i)\|\leq \gamma$ and $\tilde{\pi}_k(\tilde{t})=t_k$, for $k=1,\ldots,{\cal N}-1$. The same holds for $\omega_1,\omega_{\cal N}$, if we define $\tilde{\pi}_1(\tilde{x}):=\tilde{\pi}_1(\tilde{x}(0))$, $\tilde{\pi}_1(\tilde{t}):=0$, $\tilde{\pi}_{\cal N}(\tilde{x}):=\tilde{\pi}_{\cal N}(\tilde{x}(1))$, $\tilde{\pi}_{\cal N}(\tilde{t})=1$.

The same argument that led to eq. (\ref{taylor}) also holds, but this time the formula only applies to polynomials $h\in\p_{[0,1]}$. Namely, we have that:
\begin{equation}
\omega\left(\Pi_j\tilde{\pi}_{[0,1]}(h)\right)=\int_{[t_j,t_{j+1}]}dt\omega_j\left(\sum_{k=0}^\infty\frac{\tilde{\pi}_j(\D^nh)}{k!}(t-t_j)^k\right)=\int_{[t_j,t_{j+1}]}dt\omega_j\left(h(t,\hat{X}^{\underline{j}}(t),\tilde{\pi}_j(y_{\text{rest}}))\right),
\label{taylor_Markovian}
\end{equation}
for all $h\in\p_{[0,1]}$, where $\{\Pi_k\}_k$ are the projectors:
\begin{equation}
\Pi_k:=\Theta(t_{k+1}\tilde{\pi}_{[0,1]}(\tilde{t}))-\Theta(t_{k}\tilde{\pi}_{[0,1]}(\tilde{t})),
\end{equation}
and $\hat{X}^{\underline{j}}$ is defined through eq. (\ref{piece_sol}), but replacing $\tilde{y}$ with $\tilde{y}_{\text{rest}}$.

Next, one would have to prove an analog of Proposition \ref{prop_relations}. First, for arbitrary $j=1,\ldots,{\cal N}-1$ and $k\leq j$, we define the functions:
\begin{align}
&s^1(\bar{X},\bar{Y}_{\text{rest}})=\bar{X}(0),\nonumber\\
&s^{j+1}(\bar{X},\bar{Y}_{\text{rest}})=s(t_{j+1},s^{j}((\bar{X},\bar{Y}_{\text{rest}})),\bar{Y}_{\text{rest}};t_j).
\end{align}
We can now formulate the modified proposition:
\begin{prop}
\label{prop_relations_Markovian}
For $j=1,\ldots,{\cal N}$, the series $s^{j}(\tilde{\pi}_1(\tilde{x}),\tilde{\pi}_1(\tilde{y}_{\text{rest}}))$ converges and its entries have norm bounded by $\gamma$. Moreover, for any polynomial $\mu\in\tilde{\p}_{[0,1]}$,
\begin{equation}
\omega_1(\mu(s^{j}(\tilde{\pi}_1(\tilde{x}),\tilde{\pi}_1(\tilde{y}_{\text{rest}})),\tilde{\pi}_1(\tilde{y}_{\text{rest}})))=\omega_j\circ\tilde{\pi}_j(\mu).
\label{rel_1_j_Markovian}
\end{equation}
\end{prop}
\begin{proof}
The proposition holds for $j=1$ trivially. Assuming that the proposition holds for $j-1$, we prove that it also holds for $j$. Following mutatis mutandis the proof of Proposition \ref{prop_relations}, one shows that eq. (\ref{rel_1_j_Markovian}) holds for $j$. To complete the induction step, we must prove that $s^j(\tilde{\pi}_1(\tilde{x}),\tilde{\pi}_1(\tilde{y}_{\text{rest}}))$ is norm-bounded by $\gamma$. To do that, we set 
\begin{equation}
\mu(\tilde{x},\tilde{y}):= \nu(\tilde{x},\tilde{y}_{\text{rest}})\left(\gamma^2-\tilde{x}^2_i\right)\nu(\tilde{x},\tilde{y}_{\text{rest}})^\dagger
\end{equation}
in eq. (\ref{rel_1_j_Markovian}). The result is:
\begin{align}
&\omega_1\left(\nu(\check{X}(t_{j}),\tilde{\pi}_1(\tilde{y}_{\text{rest}}))\left(\gamma^2-\check{X}_i(t_{j})\right)\nu(\check{X}(t_{j}),\tilde{\pi}_1(\tilde{y}_{\text{rest}}))^\dagger\right)\nonumber\\
&=\omega_j\circ\tilde{\pi}_j\left(\nu(\tilde{x},\tilde{y}_{\text{rest}})\left(\gamma^2-\tilde{x}^2_i\right)\nu(\tilde{x},\tilde{y}_{\text{rest}})^\dagger\right)\geq 0,
\label{posi_Markovian}
\end{align}
where, for $t\in[t_{j-1},t_j]$, $\check{X}(t)$ is defined as:
\begin{equation}
\check{X}(t)=s(t,s^{j-1}(\tilde{\pi}_1(\tilde{x}), \tilde{\pi}_1(\tilde{y}_{\text{rest}})),\tilde{\pi}_1(\tilde{y})).
\end{equation}
Note that, by Lemma \ref{convergence_lemma}, the series above converges and $\check{X}(t)$ is the solution of the equation
\begin{align}
&\frac{dX(t)}{dt}=g(t,X(t),\tilde{\pi}_1(\tilde{y}_{\text{rest}})),\nonumber\\
&X(t_{j-1})=s^{j-1}(\tilde{\pi}_1(\tilde{x}), \tilde{\pi}_1(\tilde{y}_{\text{rest}})).
\end{align}
By the induction hypothesis and Lemma \ref{convergence_lemma}, $\check{X}(t)$ can be extended to $t\in[0,t_j]$ so that $\check{X}(t)$ satisfies the equation above for $t\in[0,t_j]$ with $\check{X}(0)=\tilde{\pi}_1(\tilde{x})$. Moreover, $\check{X}(t)$ is analytic in ${\cal Z}=\bigcup_{k=0}^j\{|t-t_k|\leq R_k\}$, which implies that there exists $R>0$ such that $\{|\mbox{Im}(t)|\leq R,-R\leq\mbox{Re}(t)\leq t_j+R\}\subset {\cal Z}$. The conditions for Lemma \ref{lemma:C_algebra_eq} are thus fulfilled, and so the $C^*$-algebras generated by $\{\check{X}(t),\tilde{\pi}_1(\tilde{y}_{\text{rest}})\}$ and $\{\tilde{\pi}_1(\tilde{x}),\tilde{\pi}_1(\tilde{y}_{\text{rest}})\}$ coincide. Since eq. (\ref{posi_Markovian}) holds for all $\nu$, we have that
\begin{align}
&\omega_1\left(\nu(\tilde{\pi}_1(\tilde{x}),\tilde{\pi}_1(\tilde{y}_{\text{rest}}))\left(\gamma^2-\check{X}_i(t_{j})\right)\nu(\tilde{\pi}_1(\tilde{x}),\tilde{\pi}_1(\tilde{y}_{\text{rest}}))^\dagger\right)\geq 0.
\end{align}
Now, $\omega_1$, the result of a GNS construction, is a cyclic state, so the above implies that
\begin{equation}
\gamma^2-s_i^j(\tilde{\pi}_1(\tilde{x}), \tilde{\pi}_1(\tilde{y}_{\text{rest}}))=\gamma^2-\check{X}_i(t_{j})\geq 0,
\end{equation}
that is, $\|s_i^j(\tilde{\pi}_1(\tilde{x}),\tilde{\pi}_1(\tilde{y}_{\text{rest}}))\|\leq\gamma$.

\end{proof}

Next, we rename $\tilde{\H}_1\to\H$ and $\omega_1\to\psi$, and define the representation $\pi:\p\to B(\H)$ through 
\begin{align}
&\pi(y_{\text{rest}}):=\tilde{\pi}_1(\tilde{y}_{\text{rest}}),\nonumber\\
&\pi(x(t)):=\hat{X}(t),\nonumber\\
\end{align}
where
\begin{equation}
\hat{X}(t):=s(t,s^{j}(\pi(x(0)),\pi(y_{\text{rest}})),\pi(y)),\mbox{ for } t_j\leq t\leq t_{j+1}.
\end{equation}
By Proposition \ref{prop_relations_Markovian}, this operator is well-defined for $t\in [0,1]$ (because all involved series converge). Moreover, by Lemma \ref{convergence_lemma}, it is a solution of the differential equation
\begin{equation}
\frac{d\hat{X}(t)}{dt}=g(t,\hat{X}(t),\pi(y_{\text{rest}})),
\end{equation}
which admits an analytic extension to ${\cal Z}=\bigcup_{k=0}^j\{|t-t_k|\leq R_k\}\supset \{|\mbox{Im}(t)|\leq R,-R\leq\mbox{Re}(t)\leq t_j+R\}$, for some $R>0$. By Lemma \ref{lemma:C_algebra_eq}, we thus have that the $C^*$-algebras $\mathbb{A}(t)$ generated by $X(t),Y{\mbox{rest}}$ satisfy $\mathbb{A}(t)=\mathbb{A}(0)$, for all $t\in[0,1]$.

Following the proof of Theorem \ref{fund_theo}, we derive the analog of eq. (\ref{connection_psi_omega}):
\begin{align}
&\int_{[0,1]}dt\psi(h(t,\pi(x(t)),\pi(y_{\text{rest}})))=\omega(\tilde{\pi}_{[0,1]}\circ h(\tilde{t},\tilde{x},\tilde{y}_{\text{rest}})).
\label{connection_psi_omega_Markovian}
\end{align}
Similarly,
\begin{align}
&\psi(h(\pi(x(0)),\pi(y_{\text{rest}})))=\omega_1(h(\tilde{\pi}_1(\tilde{x}(0)),\tilde{\pi}_1(\tilde{y}_{\text{rest}}))),\nonumber\\
&\psi(h(\pi(x(1)),\pi(y_{\text{rest}})))=\omega_N(h(\tilde{\pi}_N(\tilde{x}(1)),\tilde{\pi}_N(\tilde{y}_{\text{rest}}))),
\end{align}
where the two identities follow from the definition of $(\H,\psi,\pi)$ and Proposition \ref{prop_relations_Markovian}. This allows us to prove that $(\H,\psi,\pi)$ satisfies the last three lines of eq. (\ref{conds_theor_Markovian}).

For the operator constraints, the first three lines of eq. (\ref{conds_theor_Markovian}), let $\nu(t)$ be a polynomial, non-negative in $t\in [0,1]$. From eq. (\ref{connection_psi_omega_Markovian}), we have that, for any polynomial $h$,
\begin{align}
&\int_{[0,1]}dt\psi(\nu(t)h(\pi(x(t)),\pi(y_{\text{rest}}))p^{[0,1]}_k(t,\pi(x(t)),\pi(y_{\text{rest}}))h(\pi(x(t)),\pi(y_{\text{rest}}))^\dagger)\nonumber\\
&=\omega\circ\tilde{\pi}_{[0,1]}\left(\nu(\tilde{t})h(\tilde{x},\tilde{y}_{\text{rest}}) p^{[0,1]}_k(\tilde{t},\tilde{x},\tilde{y})h(\tilde{x},\tilde{y}_{\text{rest}})^\dagger\right)\geq 0.
\end{align}
Since this holds for all positive polynomials $\nu(t)$, it follows that, for all $t\in[0,1]$,
\begin{equation}
\psi(h(\pi(x(t),\pi(y_{\text{rest}}))p^{[0,1]}_k(t,\pi(x(t)),\pi(y_{\text{rest}}))h(\pi(x(t),\pi(y_{\text{rest}}))^\dagger)\geq 0.
\end{equation}
As this relation holds for all $h$ and $\mathbb{A}(t)=\mathbb{A}(0)$, we have that
\begin{equation}
\psi\circ\pi(h(t,x(0),y_{\text{rest}})p^{[0,1]}_k(t,x(t),y_{\text{rest}})h(t,x(0),y_{\text{rest}})^\dagger)\geq 0,
\end{equation}
for all polynomials $h$. By cyclicity of $\psi$, this implies that
\begin{equation}
\pi(p^{[0,1]}_k)\geq 0.
\end{equation}
The operator constraints
\begin{equation}
\pi(p^{\alpha}_k)\geq 0,
\end{equation}
for $\alpha\in\{0,1\}$ are proven similarly.
It follows that $(\H,\psi,\pi)$ satisfies eq. (\ref{conds_theor_Markovian}). The theorem is proven.

\end{proof}

\section{Spin systems in the thermodynamic limit}
\label{app:thermo}
In the following lines, we explain why Theorem \ref{theo:spins} holds. Our starting point is the net $(\tilde{\omega}^\kappa)_{\kappa}$, where $\tilde{\omega}^\kappa:=(\tilde{\omega}^\kappa_{[0,1]}, \tilde{b}_0^\kappa, \tilde{b}_1^\kappa)$ denotes a solution of the SDP (\ref{Markovian_NPO_rel_TI}), with $\kappa=(k_\sigma,k_t,n)$. We need to show that this net implies the existence of a Hilbert space $\H$, operators $\{\sigma_a^{(k)}:a=1,2,3\}_{k\in\Z}$ and a TI state $\psi$ such that conditions (\ref{obj_spins}) and (\ref{diff_eq_spins}) hold.

The proof is essentially identical to that of Theorem \ref{theorem:Markovian}, so we just focus on solving the theoretical hurdles that appear when one tries to adapt the methods of Appendix \ref{app:Markovian} to systems with an infinite number of operator variables.

There are four such hurdles:

\begin{itemize}
    \item The Archimedean condition is not independently satisfied, because $\tilde{\sigma}$, $\tilde{\sigma}(0)$, $\tilde{\sigma}(1)$ represent infinite sets of variables. This is not a real problem: it suffices when, for each non-commuting variable $\tilde{x}_i$, we have that the polynomial
    \begin{equation}
    K_i-\tilde{x}_i^2
    \label{almost_arch}
    \end{equation}
    is SOS (with respect to the set of operator constraints of the set of variables to which $\tilde{x}_i$ belongs) for some $K_i\in\R^+$. This is the case in our problem, since
    \begin{equation}
    1-(\tilde{\sigma}_a^{(k)})^2, 1-\tilde{\sigma}_a^{(k)}(0)^2, 1-\tilde{\sigma}_a^{(k)}(1)^2
    \end{equation}
    are operator identities and
    \begin{equation}
    1-\tilde{t}^2=\sqrt{2}(\tilde{t}-\tilde{t}^2)\sqrt{2}+(1-\tilde{t})^2.
    \end{equation}
    Condition (\ref{almost_arch}) is, indeed, enough to bound the values of $\tilde{\omega}^{\kappa}$, see the proof of \cite[Theorem 1]{NPO}. Like their finite-variable counterparts, these functionals are defined by their (bounded) values on a countable set of monomials. Thus, one can invoke the Banach-Alaoglu theorem as usual and deduce the existence of a subnet $(h(\kappa'))_{\kappa'}$ such that the limit $\lim_{\kappa'\to\infty}\tilde{\omega}^{h(\kappa')}$ exists.

    \item  The limiting functionals $\tilde{\omega}_{[0,1]},\tilde{\omega}_{0},\tilde{\omega}_1$ can only evaluate polynomials spanned by products of operator variables corresponding to the qubits ($1,2,\ldots$) (and also products of $\tilde{t}$, when we deal with $\tilde{\omega}_{[0,1]}$). To arrive at translation-invariant functionals, respectively defined for all polynomials in $\tilde{\p}_{[0,1]},\tilde{\p}_{0}, \tilde{\p}_{1}$, we thus extend the original functionals in this form:
    \begin{equation}
    \tilde{\omega}_\alpha(p):=\tilde{\omega}_\alpha(\t^{k(p)}(p)),    
    \end{equation}
    where $k(p)$ is any natural number $k$ such that $\t^{k}(p)$ only contains operators pertaining qubits $1,2,\ldots$. This definition is consistent due to LTI.

    \item Following the proof of Theorem \ref{theorem:Markovian} would require applying the GNS construction to the functionals $\tilde{\omega}_{[0,1]},\tilde{\omega}_{0}, \tilde{\omega}_{1}$ on polynomials with infinitely many variables. This is not problematic. Moreover, since the three sets of variables are countable, the constructed Hilbert space will be separable in each case.

    \item Define the polynomial 
    \begin{equation}
    g_a(\tilde{\sigma}^{(k-1)},\tilde{\sigma}^{(k)},\tilde{\sigma}^{(k+1)}):=i\tau[H(\tilde{\sigma}),\tilde{\sigma}_a^{(k)}].
    \end{equation}
    The proof of Theorem \ref{theorem:Markovian} relies on Lemma \ref{time_part_lemma}, which in turn requires that the (classical) ODE
    \begin{align}
    &\frac{d\xi_a^{(k)}(t)}{dt}= g^+_a(\xi^{(k-1)}(t),\xi^{(k)}(t),\xi^{(k+1)}(t)),\nonumber\\
    &\xi_a^{(k)}(\tau)=\gamma>0,
    \label{ODE_infty}
    \end{align}
    admits an analytic solution for some open ball around $\tau=0$ in the complex plane. In Appendix \ref{app:proof_theo}, this was dealt with by invoking the Cauchy-Kovalevskaya theorem \cite{book_diff_eqs}. However, this theorem does not apply here, since the vector of variables $(\xi_a^{(k)})_{a,k}$ is infinite.

    Nonetheless, we next show that eq. (\ref{ODE_infty}) does admit an analytic solution around $t=0$. Consider the simpler ODE:
    \begin{align}
    &\frac{d\xi_a(t)}{dt}= g^+_a(\xi(t),\xi(t),\xi(t)),a=1,2,3,\nonumber\\
    &\xi_a(0)=\gamma>0,a=1,2,3.
    \label{ODE_infty_mod}
    \end{align}
    This system of ODEs is finite dimensional, so Cauchy-Kovalevskaya applies: for some $R>0$, there exists a solution $\hat{\xi}(t)$ of eq. (\ref{ODE_infty_mod}), analytic in $|t|\leq R$. Define $\hat{\xi}_a^{(k)}(t):=\hat{\xi}_a(t)$. Then, $(\hat{\xi}_a^{(k)}(t))_{a,k}$ is an analytic solution of (\ref{ODE_infty}); note, indeed, that $\hat{\xi}_a^{(k)}(0)=\hat{\xi}_a(0)=\gamma$ and
    \begin{equation}
    \frac{d\hat{\xi}^{(k)}_a(t)}{dt}=\frac{d\hat{\xi}_a(t)}{dt} =g^+_a(\hat{\xi}(t),\hat{\xi}(t),\hat{\xi}(t))=g^+_a(\hat{\xi}^{(k-1)}(t),\hat{\xi}^{(k)}(t),\hat{\xi}^{(k+1}(t)).
    \end{equation}
    From this point on, one can follow the proof of Theorem \ref{theorem:Markovian} without further obstructions.
    
\end{itemize}

\end{appendix}

\end{document}